\documentclass[journal]{IEEEtran}
\pdfoutput=1
\IEEEoverridecommandlockouts
\IEEEoverridecommandlockouts
\usepackage{amsfonts}
\usepackage{amssymb}
\usepackage[cmex10]{amsmath}
\usepackage{graphicx}
\usepackage{subfigure}
\usepackage{verbatim}
\usepackage{cite}
\usepackage{multicol}
\usepackage{diagbox}
\usepackage{tabularx,booktabs}
\usepackage{amsthm}
\usepackage{booktabs}
\usepackage{makecell}
\usepackage{multirow}
\usepackage{xypic}
 \usepackage{enumerate}
 \usepackage{color}
 \usepackage{soul}
\usepackage[ruled,lined]{algorithm2e}
\usepackage{setspace}
\usepackage{fancyhdr}
\usepackage{algorithmic}
\usepackage[table]{xcolor}
\usepackage{diagbox}
\begin{document}
\newcommand{\be}{\begin{equation}}
\newcommand{\ee}{\end{equation}}
\newcommand{\br}{{\mbox{\boldmath{$r$}}}}
\newcommand{\bp}{{\mbox{\boldmath{$p$}}}}

\newcommand{\bn}{{\mbox{\boldmath{$n$}}}}
\newcommand{\balfa}{{\mbox{\boldmath{$\alpha$}}}}
\newcommand{\ba}{\mbox{\boldmath{$a $}}}
\newcommand{\bta}{\mbox{\boldmath{$\beta $}}}
\newcommand{\bg}{\mbox{\boldmath{$g $}}}
\newcommand{\bPsi}{\mbox{\boldmath{$\Psi $}}}
\newcommand{\bpsi}{\mbox{\boldmath{$\psi $}}}
\newcommand{\bsigma}{\mbox{\boldmath{$\Sigma $}}}
\newcommand{\bpi}{\mbox{\boldmath{$\pi$}}}
\newcommand{\bGamma}{{\bf \Gamma}}
\newcommand{\blpi}{\mbox{\boldmath{$\overline{\pi}$}}}
\newcommand{\bA}{{\bf A }}
\newcommand{\bII}{{\bf I }}
\newcommand{\bP}{{\bf P }}
\newcommand{\bX}{{\bf X }}
\newcommand{\bI}{{\bf I }}
\newcommand{\bR}{{\bf R }}
\newcommand{\bff}{{\mathbf{f}}}
\newcommand{\bZ}{{\bf Z }}
\newcommand{\bz}{{\bf z }}
\newcommand{\bx}{{\mathbf{x}}}
\newcommand{\bPi}{{\bf \Pi }}
\newcommand{\bM}{{\bf M}}
\newcommand{\bW}{{\bf W}}
\newcommand{\bU}{{\bf U}}
\newcommand{\bD}{{\bf D}}
\newcommand{\bJ}{{\bf J}}
\newcommand{\bH}{{\bf H}}
\newcommand{\bK}{{\bf K}}
\newcommand{\bm}{{\bf m}}
\newcommand{\bN}{{\bf N}}
\newcommand{\bC}{{\bf C}}
\newcommand{\bL}{{\bf L}}
\newcommand{\bF}{{\bf F}}
\newcommand{\bv}{{\bf v}}
\newcommand{\bSigma}{{\bf \Sigma}}
\newcommand{\bS}{{\bf S}}
\newcommand{\bs}{{\bf s}}
\newcommand{\bO}{{\bf O}}
\newcommand{\bQ}{{\bf Q}}
\newcommand{\btr}{{\mbox{\boldmath{$tr$}}}}
\newcommand{\bNSCM}{{\bf NSCM}}
\newcommand{\barg}{{\bf arg}}
\newcommand{\bmax}{{\bf max}}
\newcommand{\test}{\mbox{$
\begin{array}{c}
\stackrel{ \stackrel{\textstyle H_1}{\textstyle >} } { \stackrel{\textstyle <}{\textstyle H_0} }
\end{array}
$}}

\newtheorem{Def}{Definition}
\newtheorem{Pro}{Proposition}
\newtheorem{Rem}{Remark}
\newtheorem{Lem}{Lemma}
\newtheorem{Theo}{Theorem}
\newtheorem{Exa}{Example}
\newtheorem{Cor}{Corollary}
\title{Efficient dual-scale generalized Radon-Fourier transform detector family for  long time coherent integration}

\author{ Suqi Li, Yihan Wang, Bailu Wang$^*$, Giorgio Battistelli, Luigi Chisci, Guolong Cui

\thanks{
S. Li, Y. Wang and B. Wang are with the School of  Micro-electronics and Communication Engineering, Chongqing University, Chongqing, China.
G. Battistelli and Luigi are  with the Dipartimento di Ingegneria dell' Informazione (DINFO), Universit$\grave{\mbox{a}}$ degli Studi di Firenze, Firenze, Italy.
Guolong Cui is with the School of  Information and Communication Engineering, University of Electronic Science and Technology of China, Chengdu, China.
}}
\maketitle

\begin{abstract}
Long Time Coherent Integration (LTCI)   aims  to accumulate target energy through long  time integration, which is an effective method for the detection of a weak target.
However,  for a moving target, defocusing can occur due to range migration (RM) and Doppler frequency migration (DFM).
To address this issue, RM and DFM corrections are required in order to achieve a well-focused image   for the subsequent detection.
Since RM and DFM are induced by the same motion parameters,  existing approaches such as the generalized Radon-Fourier transform (GRFT) or the keystone transform (KT)-matching filter process (MFP) adopt the same search space for the motion parameters in order to eliminate both effects, thus leading to large redundancy in computation.
To this end, this paper first proposes a dual-scale decomposition of the target motion parameters, consisting of well designed coarse and fine motion parameters.
Then, utilizing this decomposition,  the joint correction of the RM and DFM effects is decoupled  into a cascade procedure, first RM correction on the coarse search space and then DFM correction on the fine search spaces.
As such,  step size of the search space can be tailored to RM and DFM corrections, respectively, thus avoiding large redundant computation effectively. The resulting algorithms are called \textit{dual-scale GRFT (DS-GRFT)} or \textit{dual-scale GRFT (DS-KT-MFP)} which provide comparable performance while achieving significant improvement in computational efficiency compared to  standard GRFT (KT-MFP).
Simulation experiments verify their effectiveness and efficiency. 
\end{abstract}
\vspace{-8pt}
\section{Introduction}
\textit{Long time coherent integration} (LTCI)   is an effective method to accumulate target energy through long  time integration \cite{LTCI-SBBR}.
However, for a  moving target, defocusing can occur due to \textit{range migration} (RM) and  \textit{Doppler frequency migration} (DFM).
Hence, RM and DFM have to be accurately compensated in order to obtain a well-focused image of a moving target in LTCI.
Essentially, the RM arises from the  coupling between fast-time and  target motions, while the DFM arises from the coupling  between  slow-time and target motions.

In the last decades, many LTCI detection methods have been proposed to compensate DFM or RM. 
As an optimal method, the \textit{generalized Radon-Fourier transform} (GRFT) \cite{GRFT} can effectively estimate  high-order motion parameters,  compensating DFM  and RM simultaneously. The original GRFT is defined in the frequency-domain. However, it is usually implemented in the time-domain   via searching the target trajectory with  arbitrarily parameterized motion model, which is in fact an efficient approximation  of the direct frequency-domain implementation  
 \cite{GRFT}. {\color{magenta}Motivated by  the superior performance of the GRFT detector, it has been  generalized to multiple channels' signal integration in \cite{MC-GRFT}, and also extended to cope with the case when the target's motion model changes during the integration time \cite{STGRFT} in recent years.}
The GRFT detector belongs to the Radon-based methods, which represents one of the broadest families of LTCI
methods. Other typical Radon-based methods include Radon fractional Fourier transform \cite{FRFT,RFRFT}, Radon linear canonical transform \cite{LCT,RLCT},
Radon Lv's distribution (LVD)  \cite{LVD,RLVD}, etc.

Combining  the keystone transform (KT) \cite{KT-SAR-Perry,KT-2,KT-SAR-Zhu}  and \textit{polynomial  phase signal} (PPS)\cite{PPS-MPT,PPS-QFM,PPS}/\textit{time-frequency} (TF) \cite{STFT,WVD,AF} methods, many enhanced LTCI algorithms, referred to as  KT-based methods  \cite{KT-MFP-Vd,KT-SAR-Perry,KT-LCT,SKT-HAF,DKT-GMTI},  have emerged over the recent years.
Specifically, KT is first used to eliminate RM and then PPS/TF is  utilized for the correction of DFM. Nevertheless, the KT and its high-order versions cannot correct RM caused by all the motion parameters simultaneously, thus  leaving a residual RM, such as the linear RM caused by baseband velocity and  the high-order RM caused by  acceleration or acceleration rate.  Hence, although KT-based methods can be implemented with lower computational load than Radon-based ones, their integration performance is worse due to the impact of the residual RM.

Another promising solution is the KT-matched filtering processing (KT-MFP) \cite{KT-MFP, KT-MFP-Three-Order} method, in which KT is applied to eliminate RM caused by the baseband velocity, then MFP is performed to remove residual RM and DFM simultaneously. Therefore, the KT-MFP detector yields similar performance as the GRFT detector, but
 is computationally  cheaper.  However, it still involves joint high-dimensional searching of ambiguous velocity, acceleration, and acceleration rate, etc.
We observe that the MFP procedure in the KT-MFP detector has the same form of the GRFT detector in frequency domain. Hence, in this paper the GRFT (in frequency-domain) and KT-MFP detectors are collectively referred to as members of the GRFT detector family.
Notice that the implementation of the GRFT detector in time-domain is  not included in the GRFT detector family since it has a different form.




In this context, our analysis shows that,  since  RM and DFM share the same motion parameters, RM and DFM correction procedures of most  LTCI detection methods are coupled in the sense that the same motion parameters have to be compensated for the two different corrections.  As a result,  the search spaces of  motion parameters for eliminating RM and DFM are  set to be the same according to the minimum step size required by RM and DFM corrections. Hereafter, we refer to this 
as \textit{coupling effect of parameter search spaces}.  In this respect, LTCI  methods  on the coupled parameter search space not only  suffer from redundant computation, but also from lack of freedom in the sense that RM and DFM have to be jointly corrected.
{\color{magenta}Even if  some  non-parametric searching-based LTCI algorithms, such as the adjacent cross correlation-based methods \cite{Non-Searching-XiaolongLi,Non-Searching-XiaolongLi-v2,BCS}, the time reversal-based method \cite{TR-SKT-FRFT} and the product scale zoom discrete chirp Fourier transform-based method \cite{SZDCFT}, can correct the RM and DFM more efficiently,  they involve nonlinear transforms, leading to  poor antinoise performance. It is worth noting that \cite{BCS,TR-SKT-FRFT} have introduced the sparse theory to the LTCI, further decreasing the computational complexity.}



In this paper, to break the coupling effect of parameter search spaces, we  first propose a dual-scale (DS) decomposition of the target motion parameters,  leading to decoupled dual-scale search spaces, i.e., coarse search space and fine  search space.
This idea is inspired by the fact that if  the step sizes are tailored to RM and DFM corrections, respectively, the computational complexity can be largely reduced.
Theoretical analysis shows that,  thanks to the proposed dual-scale decomposition, the methods of the  GRFT detector family can be reduced into a \textit{generalized inverse Fourier transform} (GIFT) process in range domain and  \textit{generalized Fourier transform} (GFT) processes in Doppler domain conditioned on the coarse motion parameter. These properties make it possible to decouple the joint correction of RM and DFM effects into a cascade procedure, first RM correction on the coarse search space and then DFM correction on the fine search spaces.  In this respect, the proposed decoupled GRFT or KT-MFP is called \textit{dual-scale GRFT} (DS-GRFT) or \textit{dual-scale KT-MFP} (DS-KT-MFP) (which are also collectively referred to as members of the DS-GRFT detector family).
Compared to the standard GRFT detector family, the DS-GRFT detector family can provide comparable performance while providing significant improvement in computational efficiency.
Simulation experiments verify the effectiveness and the efficiency of the proposed LTCI methods.

{\color{magenta}Preliminary results were published in the conference paper \cite{DS-GRFT-Conference}, in which the idea of the DS decomposition of motion parameters was first proposed and its feasibility combining the GRFT detector was preliminarily investigated. This paper provides significant improvements in terms of newly-devised algorithms, algorithm enhancement, additional analysis, mathematical proofs of propositions and extensive experimental validation.  Specifically, contributions with respect to \cite{DS-GRFT-Conference} are summarized as follows.
\begin{itemize}
 \item\textit{Theoretical analysis for two implementations of GRFT detector: } This paper  provides a theoretical analysis for detection performance of  two implementations of GRFT, i.e., the standard one, namely,  GRFT in the Frequency Domain (FD-GRFT) and GRFT in the Time Domain (TD-GRFT). Our analysis shows  that, even if the TD-GRFT has lower computational cost,  it has obvious performance loss due to  an additional phase introduced by IFFT operation (see Remark 1 and eq. (\ref{received-signal-after-compensation-v1})). This analysis fully motivates the choice of FD-GRFT when developing DS-GRFT.
 \item \textit{A complete and detailed  picture of the proposed DS-GRFT detector: }  This paper  provides a complete and detailed  picture of the proposed DS-GRFT,  including mathematical proofs of propositions, detailed algorithm implementation (i.e., pseudocode) and computational cost analysis.
  It is worth pointing out  the mathematical derivation for factorization of the GRFT detector (which is at the core of DS-GRFT) is more rigorous compared with that of the conference paper (see Propositions 1 and 2).
 \item  \textit{A newly-devised DS-KT-MFP detector: } Except for the DS-GRFT, this paper has further applied the DS decomposition to the KT-MFP detector, resulting in a newly-devised DS-KT-MFP detector.  The key idea is to derive a decoupled KT-MFP detector exploiting the proposed DS decomposition and combining the KT procedure.
 Mathematical results that support factorization of  the DS-KT-MFP are also provided along with their proofs. 

 \item  \textit{Comprehensive performance assessment in a more complicate multi-object scenario: } This paper provides a comprehensive performance assessment in terms of detection performance and execution efficiency for both the proposed DS-GRFT and DS-KT-MFP detectors, in a multi-object scenario with high order motion, by comparison with detectors of FD-GRFT, TD-GRFT and  KT-MFP.
\end{itemize}
}

%

The remaining part of this paper is organized as follows.
In Section II, the signal model is established and the RM/DFM effects are reviewed.
Then, GRFT, KT-MFP as well as their conventional implementations on the single-scale parameter search space are provided in Section III.
Section IV presents the dual-scale decomposition of motion parameters, by which  factorization of the GRFT and the KT-MFP are achieved resulting in the proposed DS-GRFT and DS-KT-MFP, respectively. Several implementation issues including pseudo-codes, computational complexity analysis and structural advantages of the proposed DS-GRFT family are provided in Section V.     In Section VI, we evaluate performance via numerical experiments under several points of view, including effectiveness of coherent integration, detection performance and computational complexity. Finally, we provide conclusions in Section VII.

\vspace{-8pt}
\section{Background}
\subsection{Receival signal model}
In this paper, the \textit{pulse compression} (PC) is implemented in the frequency domain. First, according to the stationary phase principle \cite{stationary-phase-principle}, the signal model after the received base-band signal multiplication by the reference signal in the range frequency-pulse domain can be expressed as
\begin{equation}\label{Rx-OFDM-Fre1}
{\small{\begin{split}
y(f_k,t_m)\!=\!A_1 \text{exp}\left[ -j \frac{4\pi}{c}(f_k+f_c) R(t_m) \right]\!+\!w(f_k,t_m)\\
\end{split}}}
\end{equation}
where:
$A_1$ denotes  the target complex attenuation in the frequency domain;  $t_m$  the slow time (inter-pulse sampling time), satisfying $t_m=m \cdot \text{PRT}$ for  $m=1,\cdots,M$ with $M$ the number of pulses; $\text{PRT}$ is the \textit{Pulse Repetition Time} and its corresponding \textit{Pulse Repetition Frequency} (PRF) is  $\text{PRF}=\frac{1}{\text{PRT}}$; $f_c$ is the carrier frequency; $R(t_m)$ is the instantaneous slant range between the radar and a maneuvering target;  $w(f_k,t_m)$ is the Gaussian-distributed noise in the frequency domain with zero-mean and variance $\sigma^2$.

After performing range IFFT, (\ref{Rx-OFDM-Fre1}) becomes
\begin{equation}\label{received signal}
{\small\begin{split}
&\,\,\,\,\,\,\,\,y(\tau_n,t_m)\\
&\!\!\!=\!\!A_2\text{asinc}\left[B_r\left(\tau_n\!\!-\!\!2 R(\!t_m\!)/\!c\right)\right]\!\exp\left\{j\pi B_r \left(\tau_n\!-\!2 R(\!t_m\!)/\!c\right) \!\right\}\\
&\,\,\,\,\,\,\,\,\,\,\,\,\,\,\,\,\,\,\,\,\,\,\,\,\,\,\,\,\,\,\,\,\,\,\,\,\,\,\,\,\,\,\,\times\exp\left\{-j\frac{4\pi}{\lambda}R(t_m)\right\}+w'(\tau_n,t_m)\\
\end{split}}
\end{equation}
where:  $A_2$ denotes the target complex attenuation in the time domain, i.e., $A_2=A_1 K_{\text{valid}}$; $\tau_n$  the  fast time (intra-pulse sampling time) which is a multiple of the sampling interval $T_s$, i.e., $\tau_n=nT_s$ and $T_s=1/f_s$; $c$ is the velocity of light, i.e., $c=3\times10^8$\,$\text{m/s}$; $\lambda$ is the wave length, i.e., $\lambda=c/f_c$; $w'(\tau_n,t_m)$ is the  Gaussian-distributed noise in the time domain, i.e., $w'(\tau_n,t_m)=\text{IFFT}\{w(f_k,t_m)\}$; $\text{asinc}(\cdot)$ is the \textit{aliased sinc} (asinc) function \cite{asinc1,SA-FIAA} (also called Dirichlet  \cite{asinc2} or periodic sinc function), i.e.,
\begin{equation}\label{asinc}
{\small\begin{split}
\text{asinc}(B_r \tau_n)=\frac{\sin(\pi B_r \tau_n)}{K_{\text{valid}}\sin(\pi \Delta f \tau_n)}
\end{split}}
\end{equation}
with $\Delta f$ being the interval of frequency bins, $K_{\text{valid}}$ the valid number of frequency bins, and $B_r$  the signal bandwidth, i.e., $B_r=K_{\text{valid}} \, \Delta f$.

\begin{Rem}
 Note that the signal after performing range IFFT  in (\ref{received signal}) accounts for the fact that (\ref{Rx-OFDM-Fre1}) is a discrete signal in the frequency domain. Hence, the expression in (\ref{received signal}) is different from the usual one  derived from IFFT of a continuous signal in the frequency domain, i.e., \cite{GRFT}
  \begin{equation}\label{received signal_usual}
{\small\begin{split}
&\,\,\,y(\tau_n,t_m)\\
&\!\!\!\!\!\!=\!\!A_2\text{sinc}[B_r(\tau_n\!\!-\!\!2 R(\!t_m\!)/\!c)]\!\exp\{\!-\!j\frac{4\pi}{\lambda}R(t_m)\}\!+\!\!w'(\tau_n,\!t_m)\\
\end{split}}
\end{equation}
where
\begin{equation}\label{sinc}
{\small\begin{split}
\text{sinc}(B_r \tau_n)=\frac{\sin(\pi B_r \tau_n)}{\pi B_r \tau_n}.
\end{split}}
\end{equation}
Specifically, there are two differences between the expressions of the range-pulse domain signal given in (\ref{received signal}) and (\ref{received signal_usual}) (which is the same as  (9) in \cite{GRFT}). The first is an additional phase term
 $\exp\{j\pi B_r (\tau_n\!-\!2(c_0+{\sum}_{p=1}^{P} c_p t_m^p)/c)) \}$ in (\ref{received signal}); the other is that the envelope of (\ref{received signal}) is weighted by the $\text{asinc}(\cdot)$ function while in (\ref{received signal_usual}) it is weighted by
 the $\text{sinc}(\cdot)$ function. The relationship between $\text{asinc}(\cdot)$ and $\text{sinc}(\cdot)$ functions is given by \cite{asinc1}
\begin{equation}\label{Asinc_and_Sinc}
{\small\begin{split}
\lim_{\substack{ \Delta f \rightarrow 0;\\K_{\text{valid}}\Delta f=B_r}} \text{asinc}(B_r \tau_n)=\text{sinc}( B_r \tau_n).
\end{split}}
\end{equation}
\end{Rem}
\subsection{RM and DFM effects}
Assume that there is a moving target with initial slant range $c_0$ at $t_m=0$.
Then, $R(t_m)$ can be  expressed as
\begin{equation}\label{slant range}
{\small\begin{split}
R(t_m)=c_0+{\sum}_{p=1}^{P} c_p t_m^p \\
\end{split}}
\end{equation}
where   $c_1,\cdots,c_P$ are unknown motion parameters to be estimated. Specifically, when $P=3$, $c_1$, $c_2$ and $c_3$ represent the radial components of, respectively,  target initial velocity, acceleration and acceleration rate.

Substituting (\ref{slant range}) into (\ref{received signal}) yields
\begin{equation}\label{Rx_OFDM}
{\small{\begin{split}
&y(\tau_n,t_m)=A'_2\exp\left\{-j\frac{4\pi}{\lambda}\left( c_1+{\sum}_{p=2}^{P} c_p t_m^p\right)\right\}\\
 &\,\,\,\,\,\,\,\,\,\,\,\,\,\,\,\,\,\,\,\,\,\,\,\,\,\,\,\,\,\times \text{asinc}[B_r(\tau_n-2(c_0+{\sum}_{p=1}^{P} c_p t_m^p)/c)]\\
 &\times \exp\{j\pi B_r(\tau_n-2(c_0+{\sum}_{p=1}^{P} c_p t_m^p)/c) \}\!+\!w'(\tau,t_m) \, .
\end{split}}}
\end{equation}
where $A'_2=A_2\exp\left\{-j\frac{4\pi c_0}{\lambda}\right\}$. It is seen from (\ref{Rx_OFDM}) that the terms $2c_p t_m^p/c$ ($p=1,\cdots,P$) may result in the RM effect. Specifically, the $t_m$- and $t_m^2$-terms may cause linear range walk and quadratic range curvature respectively. {\color{magenta}{Additionally, each term $2c_p t_m^p /\lambda$ ($p\geq2$) may lead to $p$-order DFM.}} Both RM and DFM effects are major factors leading to target energy dispersion if the corresponding motion parameters are not effectively compensated.

\section{GRFT detector family and its coupling effect of parameter search spaces}
In this paper, we consider  the GRFT  and KT-MFP detectors as reference  LTCI algorithms, due to their superior detection performance. As it will be shown later in this section, the KT-MFP has the same form of the GRFT detector, thus both are collectively referred to as the GRFT detector family.
\subsection{Two implementations of GRFT detector and performance analysis}
For a  given vector of motion parameter variables $\tilde{\mathbf{c}}=[\tilde{c}_0,\cdots,\tilde{c}_P]$, the test statistics of the   GRFT detector  in the frequency-domain can be  expressed in the following form
\begin{equation}\label{standard-GRFT-FD}
\!\!\!\!{\small\begin{split}
&\text{GRFT}_{y(f_k,t_m)}(\tilde{\mathbf{c}})\\
=&\sum_{m=1}^M \!\sum_{k=1}^K  y(f_k,\!t_m) \overline{H}(f_k,\!t_m;\tilde{\mathbf{c}})\exp\!\!\left[j\frac{4\pi}{c}\! f_k \tilde{c}_0 \right]\!\exp\!\!\left[j\frac{4\pi}{\lambda} \tilde{c}_1 t_m \right]
\end{split}}
\end{equation}
where: $\tilde{c}_0$ denotes the target slant range;  $\tilde{c}_p$ denotes the $p$th- order motion parameter;  $\overline{H}(f_k,t_m;\tilde{\mathbf{c}})$ is the   \textit{Matching Filter} (MF) function defined by
\begin{equation}\label{H-GRFT}
{\small\begin{split}
\overline{H}(f_k,t_m;\tilde{\mathbf{c}})\triangleq\overline{H}_{\text{RM}}(f_k,t_m;\tilde{\mathbf{c}}) \, \overline{H}_{\text{DFM}}(t_m;\tilde{\mathbf{c}})
\end{split}}
\end{equation}
$\overline{H}_{\text{RM}}(f_k,t_m;\tilde{\mathbf{c}})$  and $\overline{H}_{\text{DFM}}(t_m;\tilde{\mathbf{c}})$  being the matching coefficients with respect to RM and DFM, respectively, i.e.
\begin{align}
\label{H-RM}
\overline{H}_{\text{RM}}(f_k,t_m;\tilde{\mathbf{c}})&\small\triangleq\exp \left[j\frac{4\pi}{c} f_k \left({\sum}_{p=1}^{P} \tilde{c}_p t_m^p\right)\right],\\
\label{H-DFM}
\overline{H}_{\text{DFM}}(t_m;\tilde{\mathbf{c}})&\small\triangleq\exp \left[j\frac{4\pi}{\lambda}  \left({\sum}_{p=2}^{P} \tilde{c}_p t_m^p\right)\right].
\end{align}
In particular, when  $P=1$, the GRFT detector reduces to the RFT detector which only considers  RM caused by target radial  velocity \cite{RFT}.

{\color{magenta}{Generally, there are two implementations for the GRFT detector, namely  TD-GRFT and  FD-GRFT.
Specifically, the FD-GRFT implements the original test statistics given in (\ref{standard-GRFT-FD}) directly, achieving superior detection performance at the price of high
computational cost. For the sake of computational efficiency, an alternative is the TD-GRFT which  uses the test statistics in the time domain by searching the motion trajectory with the arbitrarily parameterized motion model. Even if the TD-GRFT can  decrease the computational cost, it has obvious energy loss, possibly leading to  degradation of the detection performance as we will analyze later in this subsection.}}
\subsubsection{GRFT detector in the time domain}
The TD-GRFT detector is an approximation of its original frequency-domain representation (\ref{standard-GRFT-FD}) according to \cite{GRFT}. More specifically, its  test statistics  is given by \cite{GRFT}
\begin{equation}\label{standard-GRFT-TD}
{\small{\begin{split}
&\text{GRFT}_{y(\tau_n,t_m)}(\tilde{\mathbf{c}})=\!\!\sum_{m=1}^M  y(\tau_{n(\tilde{\mathbf{c}},t_m)},\!t_m) \, \overline{H}_{\text{DFM}}(t_m;\tilde{\mathbf{c}})\exp\!\!\left[j\frac{4\pi}{\lambda} \tilde{c}_1 t_m \right]
\end{split}}}
\end{equation}
where $n(\tilde{\mathbf{c}},t_m)\triangleq\text{round}\left\{ \frac{2{\sum}_{p=0}^{P} \tilde{c}_p t_m^p}{cT_s} \right\}$.

  Define $y'(\tau_{n(\tilde{\mathbf{c}},t_m)},t_m)\triangleq y(\tau_{n(\tilde{\mathbf{c}},t_m)},\!t_m) \, \overline{H}_{\text{DFM}}(t_m;\tilde{\mathbf{c}})$. When $|\tilde c_p-c_p| \rightarrow 0$, $\forall  p=2,\cdots, P$,
$y'(\tau_{n(\tilde{\mathbf{c}},t_m)},t_m)$ yields:
\begin{equation}\label{received-signal-after-compensation-v1}
{\small\begin{split}
&y'(\tau_{n(\tilde{\mathbf{c}},t_m)},t_m)\\
=&A'_2\text{asinc}[B_r\Delta \tau(\tilde{\mathbf{c}},t_m)]\exp[j\pi B_r \Delta \tau(\tilde{\mathbf{c}},t_m) ]\\
&\,\,\,\,\,\,\,\,\,\,\,\,\,\,\,\,\,\,\,\,\,\,\,\,\,\,\,\,\times\exp \left[ \!-j4\pi c_1 t_m/\lambda  \right]\!+\!w'(\tau_{n(\tilde{\mathbf{c}},t_m)},t_m)
\end{split}}
\end{equation}
where $\Delta \tau(\tilde{\mathbf{c}},t_m)$ denotes the discretization error, i.e.,
\begin{equation}\label{Delta-tau}
{\small\begin{split}
\Delta \tau(\tilde{\mathbf{c}},t_m)\!=\!\tau_{\text{round}\{2{\sum}_{p=0}^{P} \tilde{c}_p t_m^p/T_sc\}}\!-\!2{\sum}_{p=0}^{P} c_p t_m^p/c.
\end{split}}
\end{equation}

Eq. (\ref{received-signal-after-compensation-v1}) shows that
the phase of $y'(\tau_{n(\tilde{\mathbf{c}},t_m)},t_m)$ is modulated not only by the target velocity, i.e., $\exp \left[ -j4\pi c_1/\lambda t_m \right]$,  but also by the discretization error $\Delta \tau(\tilde{\mathbf{c}},t_m)$, i.e., $\exp[j\pi B_r \Delta \tau(\tilde{\mathbf{c}},t_m) ]$.

Specifically, when $|\tilde c_p-c_p| \rightarrow 0, \forall p=0,\cdots, P$,  we have
\begin{equation}
{\small\begin{split}|\Delta \tau(\tilde{\mathbf{c}},t_m)|<0.5T_s,
\end{split}}
\end{equation}
 then for $m=1,\cdots,M$, the phase of $y'(\tau_{n(\tilde{\mathbf{c}},t_m)},t_m)$ caused by the term $\exp[j\pi B_r \Delta \tau(\tilde{\mathbf{c}},t_m) ]$ satisfies
\begin{equation}\label{angle_of_addtion_phase}
{\small\begin{split}
&|\angle (\exp[j\pi B_r \Delta \tau(\tilde{\mathbf{c}},t_m) ])| \leq\frac{B_r 360}{f_s 4} \leq 90(\deg),
\end{split}}
\end{equation}
applying $f_s\geq B_r$ (according to Nyquist Sampling Theorem).

Accordingly, since $|\angle (\exp[j\pi B_r \Delta \tau(\tilde{\mathbf{c}},t_m) ])|$ varies from $-90 (\deg)$ to $90 (\deg)$,  the phase alignment of $y'(\tau_{n(\tilde{\mathbf{c}},t_m)},t_m)$ for $m=1,\cdots,M$ cannot be guaranteed by $\exp \left[ \!-j4\pi c_1 t_m/\lambda  \right]$, thus the coherent integration  in (\ref{standard-GRFT-TD}), i.e., $\sum_{m=1}^M y'(\tau_{n(\tilde{\mathbf{c}},t_m)},t_m)\exp\!\!\left[j\frac{4\pi}{\lambda} \tilde{c}_1 t_m \right]$,
can incur in energy loss, possibly leading to  degradation of the detection performance.
\subsubsection{GRFT detector in frequency domain}
The direct  implementation of the GRFT detector is  in frequency domain  according to (\ref{standard-GRFT-FD}). By replacing the summations over frequency bins with $\text{IFFTs}$ in (\ref{standard-GRFT-FD}),  the test statistics can be further rewritten as
\begin{equation}\label{standard-GRFT-FD-implemeatation}
{\small\begin{split}
&\text{GRFT}_{y(f_k,t_m)}(\tilde{\mathbf{c}})\\
=&\text{IFFT}\!\left\{\!\sum_{m=1}^M  y(f_k,\!t_m) \!\overline{H}(f_k,\!t_m;\tilde{\mathbf{c}})\exp\!\!\left[j\frac{4\pi}{\lambda}\!  \tilde{c}_1 t_m \right]\right\}.
\end{split}}
\end{equation}
Furthermore, by exchanging IFFT and summation, the test statistics can be rewritten as
\begin{equation}\label{standard-GRFT-FD-implemeatation}
{\small\begin{split}
&\text{GRFT}_{y(f_k,t_m)}(\tilde{\mathbf{c}})\\
=&\sum_{m=1}^M\text{IFFT}\!\left\{\!  y(f_k,\!t_m) \!\overline{H}(f_k,\!t_m;\tilde{\mathbf{c}})\right\}\exp\!\!\left[j\frac{4\pi}{\lambda}\!  \tilde{c}_1 t_m \right]\\
=&\sum_{m=1}^M A'_2\text{asinc}[B_r \Delta \tau'(\tilde{\mathbf{c}},t_m)]\exp\{j\pi B_r \Delta \tau'(\tilde{\mathbf{c}},t_m) \}\\
&\,\,\,\,\,\,\,\,\,\,\,\,\,\,\times\exp\!\!\left[j\frac{4\pi}{\lambda}\!  {\sum}_{p=1}^P (\tilde{c}_p-c_p) t_m^p \right]+\overline{w}(\tau_n,t_m)
\end{split}}
\end{equation}
where
\begin{align}\label{discretization-error-FD-v1}
\overline{w}(\tau_n,t_m)\!\!=&\sum_{m=1}^M\text{IFFT}\!\!\left\{w(f_k,t_m) \overline{H}(f_k,t_m;\tilde{\mathbf{c}})\right\}\exp\left[j\frac{4\pi}{\lambda}\!  \tilde{c}_1 t_m \right]\\
\Delta \tau'(\tilde{\mathbf{c}},t_m)=&\tau_{n}-2 \left(c_0+{\sum}_{p=1}^P (c_p-\tilde{c}_p)t_m^p\right)/c \,
\end{align}
with $n=\text{round}(2\tilde{c}_0/cT_s)$.

When $|\tilde c_p-c_p| \rightarrow 0 $, $\forall p=0,\cdots, P$, we have
\begin{equation}\label{discretization-error-FD-v2}
{\small\begin{split}
\Delta \tau'(\tilde{\mathbf{c}},t_m)\approx \tau_{\text{round}(2\tilde{c}_0/cT_s)}-2c_0/c=\Delta \tau'(\tilde{c}_0).
\end{split}}
\end{equation}
Then, the test statistics of $\text{GRFT}_{y(f_k,t_m)}(\tilde{\mathbf{c}})$   reduces to
\begin{equation}\label{standard-GRFT-FD-implemeatation}
{\small\begin{split}
&\!\text{GRFT}_{y(f_k,t_m)}(\tilde{\mathbf{c}})\\
\!\approx\!&A_3\text{asinc}[\!B_r \Delta \tau'(\tilde{c}_0)]\text{asinc}[2\text{PRF}(\tilde{c}_1\!-\!c_1\!)/\lambda]\!+\!\overline{w}(\tau_n,\!t_m\!)
\end{split}}
\end{equation}
where 
$A_3=A'_2 M\exp\{j\pi B_r \Delta \tau'(\tilde{c}_0) \exp\{j 2\pi\text{PRF}(\tilde{c}_1-c_1)/\lambda\}$.
Obviously, the target energy is fully focused at the target motion variable $\tilde{\mathbf{c}}$.

Hence, in this paper we focus on the GRFT detector in the frequency domain.
\subsection{KT-MFP detector and connection to  the GRFT detector}
For a slowly-moving target, it is possible that its velocity  satisfies   $c_1\in [-V_a/2,V_a/2]$ where $V_a$ denotes the blind velocity.
However, for a fast-moving target, its velocity may exceed the blind velocity $V_a$, and thus can be decomposed as
\begin{equation}\label{velocity_Amb}
{\small\begin{split}
c_1=c'_1+q V_{a}=c'_1+q \frac{\lambda \cdot \text{PRF}}{2}
\end{split}}
\end{equation}
where $c'_1$ is defined as the baseband velocity satisfying $c'_1\in [-V_a/2,V_a/2]$,  and {\color{magenta}$q$ is the  folding factor corresponding to the Doppler center ambiguity number}.

According to   (\ref{velocity_Amb}),   the target velocity is split into baseband velocity $c'_1$ and ambiguity velocity $q V_a$. Since $e^{-j\frac{4\pi}{\lambda} q V_{a} t_m}=1$, substituting (\ref{velocity_Amb}) into (\ref{Rx-OFDM-Fre1}), $y(f_k,t_m)$ can be rewritten as,
\begin{equation}\label{Rx-OFDM-Fre2}
{\small\begin{split}
&\,\,\,\,\,\,y(f_k,t_m)\\
&=A_1\exp[-j\frac{4\pi}{\lambda}c_0]\\
             &\times \text{exp}\left[-j\frac{4\pi}{c}f_k(c_0+c'_1 t_m+ q V_{a} t_m+{\sum}_{p=1}^P c_p t_m^p)\right]\\
             &\times\text{exp}\left[ -j\frac{4\pi}{\lambda} (c'_1 t_m+ {\sum}_{p=2}^P c_p t_m^p)\right]+w(f_k,t_m) .
\end{split}}
\end{equation}

KT is a linear transform that can simultaneously eliminate the effects of the linear RM for all moving targets regardless of their unknown velocities. By exploiting KT,  \cite{KT-MFP}  and \cite{KT-MFP-Three-Order} successively proposed KT-MFP detectors to  eliminate RM and DFM when $P=2$ and $P=3$. In this paper, we only consider the KT-MFP detector with $P=2$ for ease of presentation. Hence, after KT, $y(f_k,t_m)$ in (\ref{Rx-OFDM-Fre2}) is simplified as
\begin{equation}\label{Echo-after-KT}
{\small\begin{split}
&y_{\text{KT}}(f_k,t_m)\\
=&A_1\exp[-j\frac{4\pi}{\lambda}c_0]\\
 &\times \text{exp}\left[-j\frac{4\pi}{c}f_k (c_0+\overline{f} q V_a t_m -c_2 t_m^2)\right]\\
 &\times\text{exp}\left[-j\frac{4\pi}{\lambda} (c_1' t_m+c_2 t_m^2)\right]+\overline{w}(f_k,t_m)\\
\end{split}}
\end{equation}
where  $\overline w(f_k,t_m)$ is the Gaussian noise after KT in the frequency domain and $\overline{f}=\frac{f_c}{f_k+f_c}$. In addition, $\overline{f}\approx1-\frac{f_k}{f_c}$ according to the Taylor series expansion.


 Then, for a  given vector of motion parameter variables $\tilde{\mathbf{c}}=[\tilde{c}_0,\tilde{q},\tilde{c}'_1,\cdots,\tilde{c}_P]$, the MF coefficients of KT-MFP are constructed as \cite{KT-MFP}
\begin{equation}\label{H-MFP}
{\small\begin{split}
\overline{H}_{\text{KT}}(f_k,t_m;\tilde{\mathbf{c}})=\overline{H}_{\text{KT},\text{RM}}(f_k,t_m;\tilde{\mathbf{c}}) \, \overline{H}_{\text{DFM}}(t_m;\tilde{\mathbf{c}})
\end{split}}
\end{equation}
where
\begin{equation}\label{H-MFP-RM}
{\small\begin{split}
\overline{H}_{\text{KT},\text{RM}}(f_k,t_m;\tilde{\mathbf{c}})\!=\!\exp \!\left[j\frac{4\pi}{c}f_k\!\left(\overline{f}\tilde{q} V_a t_m \!-\!\tilde{c}_2 t_m^2\right)\right].
\end{split}}
\end{equation}
Finally, the KT-MFP detector is given as follows
\begin{equation}\label{KT-GRFT}
{\small\begin{split}
&\,\,\,\,\,\,\,\,\text{KT-MFP}_{y(f_k,t_m)}(\tilde{\mathbf{c}})\\
&=\text{FFT}\left\{\text{IFFT}\left\{y_{\text{KT}}(f_k,t_m)\overline{H}_{\text{KT}}(f_k,t_m;\tilde{\mathbf{c}})\right\}\right\}\, .
\end{split}}
\end{equation}
By replacing IFFTs with summations over frequency bins  and FFTs with summations over slow times in (\ref{KT-GRFT}), the test statistics can be
further rewritten as
\begin{equation}\label{KT-GRFT-GRFT}
{\small\begin{split}
&\,\,\,\,\,\,\,\,\text{KT-MFP}_{y(f_k,t_m)}(\tilde{\mathbf{c}})\\
&=\sum_{m=1}^M \!\sum_{k=1}^K  y_{\text{KT}}(f_k,\!t_m) \, \overline{H}_{\text{KT}}(f_k,t_m;\tilde{\mathbf{c}}) \\
&\,\,\,\,\,\,\,\,\,\,\,\,\,\,\,\,\,\,\,\,\,\,\,\,\,\,\,\,\,\,\,\,\,\,\,\,\,\,\,\,\,\,\,\,\,\,\,\,\,\,\,\times\exp\!\!\left[j\frac{4\pi}{c}\! f_k \tilde{c}_0 \right]\!\exp\!\!\left[j\frac{4\pi}{\lambda} \tilde{c}_1' t_m \right]\\
&=\text{GRFT}_{y_{\text{KT}}(f_k,t_m)}(\tilde{\mathbf{c}})\, .
\end{split}}
\end{equation}
Notice from (\ref{KT-GRFT-GRFT})  that the MFP after KT is in the same form of (\ref{standard-GRFT-FD}), thus the KT-MFP can  also be regarded as a special GRFT detector. Similar considerations also hold for KT-MFP with $P=3$.

Hereafter, we collectively refer to the RFT, FD-GRFT, and KT-MFP detectors  for $P=2$ and $P=3$ as members of standard GRFT detector family.
\subsection{Detection on single-scale search space}
Based on (\ref{standard-GRFT-FD}), the hypothesis test for the GRFT detector family is given by
\begin{equation}\label{LRT}
{\small\begin{split}
|\text{GRFT}_{y(f_k,t_m)}(\tilde{\mathbf{c}})|=T_{\text{LRT}}(\tilde{\mathbf{c}}) \mathop{\lessgtr}\limits_{H_1}^{H_0}\gamma_{\text{LRT}}, \,\tilde{\mathbf{c}} \in \mathbb{C},
\end{split}}
\end{equation}
where: $\gamma_{\text{LRT}}$ is a suitable threshold;
$\mathbb{C}\triangleq\mathbb{C}_0\times \cdots \times \mathbb{C}_P$ and $\mathbb{C}_p$ is the search space of the $p$-th order motion.
Let  $\Delta c_p$ denote the step size for the $p$-th order motion parameter, then the search space $\mathbb{C}_p$ is  defined as
 \begin{equation}\label{delta-Beta-1}
{\small \begin{split}
\mathbb{C}_p\triangleq [c_{p,\text{min}}:\Delta c_p:c_{p,\text{max}}], \,\,p=1,\cdots, P.
\end{split}}
\end{equation}
Assuming that the noise is white and Gaussian, and given the false alarm probability $P_{\text{FA}}$, the detection threshold    can be determined  as
\begin{equation}
{\small\begin{split}\gamma_{\text{LRT}}=\sqrt{-M \, K \, \sigma^2 \, \ln P_{\text{FA}}} \, .
\end{split}}
\end{equation}

 When  $\tilde{\bf c}=[\tilde{c}_1,\cdots,\tilde{c}_P]$ satisfies $|\text{GRFT}_{y(f_k,t_m)}(\tilde{\mathbf{c}})|>\gamma_{\text{LRT}}$, $\tilde{\bf c}$ represents the estimates of the true motion parameter $\mathbf{c}=[c_1,\cdots,c_P]$. Then, we have
\begin{equation}
{\small\begin{split}\tilde{c}_p=\text{round}\left(\frac{c_p-c_{p,\min}}{\Delta c_p}\right)\Delta c_p+c_{p,\min} \, .
\end{split}}\end{equation}
Next, the following holds:
\begin{equation}
{\small\begin{split}\left| c_p-\tilde{c}_p\right| \leq \Delta c_p /2 \, .
\end{split}}
\end{equation}
To effectively focus the energy of  a moving target, after correction the  RM quantity $\left|\sum_{p=1}^P (c_p-\tilde c_p) t_m^p\right|$ during the coherent integration time must be less than a half of range resolution cell, while the  DFM quantity $\left|\sum_{p=1}^{P} (c_p-\tilde c_p)  t_m^{p-1}\right|$ must also be limited within a half of Doppler resolution cell. Hence, according to the above restrictions, one has
\begin{align}
\label{RM-DeltaR}
\left | \sum_{p=1}^P (c_p-\tilde{c}_p) t_m^p\right | &\leq \sum_{p=1}^P  \frac{\Delta c_p T^p}{2} \leq \frac{\Delta R}{2}\\
\label{RM-Deltafd}
\left | \sum_{p=1}^P \frac{2(c_p-\tilde{c}_p) t_m^{p-1}}{\lambda}\right | &\leq \sum_{p=1}^P \frac{ \Delta c_p T^{p-1}}{\lambda} \leq \frac{\Delta f_d}{2}
\end{align}
where $\Delta f_d$ denotes the Doppler resolution and $\Delta R$ the range bin. According to the radar principles, the  Doppler resolution $\Delta f_d$ is inversely proportional to the coherent integration time $T$, i.e.,
\begin{equation}\label{delta-V}
{\small\begin{split}\Delta f_d=\frac{1}{T},
\end{split}}
\end{equation}
with $T\!\!=\!\!M \cdot\text{PRT}$, while the range bin $\Delta R$ is inversely proportional to the sampling frequency $f_s$, i.e.,
\begin{equation}\label{delta-R}
{\small\begin{split}\Delta R=\frac{c}{2 f_s}.
\end{split}}
\end{equation}
Therefore, RM and DFM do not occur if
\begin{equation}\label{delta-Beta-1}
{\small\begin{split}\Delta c_p = \min \left(\Delta_{\text{RM},p},\Delta_{\text{DFM},p}\right)
\end{split}}
\end{equation}
where
\begin{align}
\Delta_{\text{RM},p}\triangleq&\alpha_p  c / (2 f_s T^{p}) \\
\Delta_{\text{DFM},p}\triangleq&\alpha_p c/ (2 f_c T^{p})
\end{align}
with $0 \leq\alpha_p\leq 1, p=1,\cdots,P$, satisfying $\sum_{p=1}^P\alpha_p=1$.
Specifically, when $p=3$,  $\Delta c_1$, $\Delta c_2$ and $\Delta c_3$ denote the step size for velocity, acceleration and acceleration rate, respectively.

Since $f_c$ is generally much larger than $f_s$,  the step sizes $\Delta c_p$  are usually set as
\begin{equation}\label{delta-Beta-1}
{\small\begin{split}\Delta c_p= \Delta_{\text{DFM},p},\, p=1,\cdots,P \, .
\end{split}}
\end{equation}
%
Consequently, both RM and DFM elimination procedures are performed on the same motion parameter search space.  This  phenomenon is referred to as \textit{the coupling effect of parameter search spaces}.
Moreover, the size of  the search space is determined by $\Delta_{\text{DFM}}$ rather than  $\Delta_{\text{RM}}$ (since $\Delta_{\text{DFM}}$ is generally much smaller than $\Delta_{\text{RM}}$).

As for  the KT-MFP detector, even if  the velocity searching for RM and DFM corrections needs not be performed on the same search space,  searches of  the higher order motion parameters still need to be performed at the unified step size $\Delta_{\text{DFM}}$.
As a result,  \emph{\textbf{searching the maxima of the outputs of  the GRFT detector family is on a single-scale space $\mathbb{C}$ which is determined by the unified  fine step size $\Delta c_p$}}, making  the GRFT detector family computationally intensive.
\section{Dual-scale GRFT detector family}
According to (\ref{RM-DeltaR}) and (\ref{RM-Deltafd}), the required step size for eliminating RM and DFM  are much different.
In fact, the sensitivity of DFM caused by the accelerated velocity is $\frac{f_c}{f_s}$ times that of the RM, where $f_c$ is usually much bigger than  $f_s$, e.g.,  $\frac{f_c}{f_s}>10$.  To address the high computational complexity issue, the idea is that  if  the implementation of  the GRFT detector family can  first eliminate the RM effect with the required step size for the range domain, and then perform a fine tuning at a higher resolution to eliminate the DFM effect,  the computational complexity could be significantly reduced.
Inspired by this,  a  more efficient implementation of the GRFT  detector family will be developed in the following sections.

\subsection{Generalized PC and MTD procedure}
For the subsequent developments, in this section we first define GFT and GIFT procedures.
{\color{magenta}\begin{Def}\label{Definition-GFT}
The GFT procedure among the echoes $y({\color{magenta}{\tau_n}},t_m), m=1,\cdots,M$, in the pulse domain is defined as follows
\begin{equation}\label{Definition-GFT}
{\small{\begin{split}
&\emph{GFT}_{y({\color{magenta}{\tau_n}},t_m)}(\tilde{\mathbf{c}})\\
=&{\color{magenta}{\sum_{m=1}^{M}}} y({\color{magenta}{\tau_n}},t_m) \overline{H}_{\emph{DFM}}(t_m;\tilde{\mathbf{c}}) \exp\left[j\frac{4\pi}{\lambda} \tilde{c}'_1 t_m \right]\\
=&\emph{FFT}\left\{y({\color{magenta}{\tau_n}},t_m)\overline{H}_{\emph{DFM}}(t_m;\tilde{\mathbf{c}}) \right\}
\end{split}}}
\end{equation}
where $\tilde{c}_1'$ denotes the baseband velocity of $\tilde{c}_1$.
\end{Def}}
\begin{Rem}
 Compared to the traditional \textit{moving target detection} (MTD), the GFT procedure not only compensates  the phase difference caused by velocity, but also by acceleration and acceleration rate. Hence, the GFT can be regarded as a generalized MTD method. Notably, the GFT can achieve better LTCI performance for a maneuvering target with complex motion.
\end{Rem}
\begin{Def}\label{Definition-GIFT}
 The GIFT procedure among the echoes $y(f_k,t_m)$ of the form (\ref{Rx-OFDM-Fre1}), $m=1,\cdots, M$, in the range-frequency domain is defined as follows
\begin{equation}\label{Definition-GIFT}
{\small{\begin{split}
&\emph{GIFT}_{y(f_k,t_m)}(\tilde{\mathbf{c}})\\
=&{\sum}_{k=1}^{K}  y(f_k,t_m)\overline{H}_{\emph{RM}}(f_k,t_m;\tilde{\mathbf{c}})\exp\left[j\frac{4\pi}{c}\! f_k \tilde{c}_0 \right]  \\
=&\emph{IFFT}\left\{y(f_k,t_m)\overline{H}_{\emph{RM}}(f_k,t_m;\tilde{\mathbf{c}}) \right\} \, .
\end{split}}}
\end{equation}
\end{Def}
Notice that $\overline{H}_{\text{RM}}(f_k,t_m;\tilde{\mathbf{c}})$ can  be replaced by $\overline{H}_{\text{RM,KT}}(f_k,t_m;\tilde{\mathbf{c}})$ when the inputs are $y_{\text{KT}}(f_k,t_m)$ which are the outputs of the KT on echoes $y(f_k,t_m)$.
\begin{Rem}
Compared to the traditional PC, the GIFT  performs an extra procedure for the compensation of the envelope shift  caused by velocity, acceleration and acceleration rate. Hence, the GIFT can be regarded as a generalized PC method implemented in the range frequency-pulse domain. After the GIFT, the target will fall into the same range bin among multiple pulses.
\end{Rem}

Further, for notational simplicity, we  define a modified GIFT i.e.,
\begin{equation}\label{Definition-modified-GIFT}
{\small\begin{split}
&\text{mGIFT}_{y(f_k,t_m)}(\tilde{\mathbf{c}})\\
=&\text{IFFT}\left\{y(f_k,t_m) \, \overline{H}_{\text{RM}}(f_k,t_m; \tilde{\mathbf{c}})  \right\} H_{\text{pre-DFM}}(t_m;\tilde{\mathbf{c}})
\end{split}}
\end{equation}
where
\begin{equation}\label{pre-DFM}
{\small\begin{split}
H_{\text{pre-DFM}}(t_m;\tilde{\mathbf{c}})=\!\!\left\{\begin{array}{lc}
\!\!\exp\!\!\left[j\frac{4\pi}{\lambda}\!  \sum_{p=1}^P \tilde{c}_{p} t_m^p \right],& \!\!\text{GRFT }\\
\!\!\exp\!\!\left[j\frac{4\pi}{\lambda}\!  \tilde{c}_{2} t_m^2 \right],&\!\!\text{KT-MFP}.
\end{array}
\!\!\right.
\end{split}}
\end{equation}
The modified GIFT differs from the original one only by a phase $H_{\text{pre-DFM}}(t_m;\tilde{\mathbf{c}})$. However, this difference does not affect the IFFT operation.
Indeed, this phase difference can be treated as a pre-compensation for the DFM effects.
\subsection{Dual-scale decomposition of motion parameters }
For the purpose of decoupling motion parameters and enhance the freedom of range compensation and Doppler compensation, we propose a  dual-scale decomposition of  the compensation parameters as follows,
\begin{equation}\label{Partition-Parameter}
{\small\begin{split}
\tilde{c}_{p}&=\tilde{c}_{p,c}+\tilde{c}_{p,f}, \,\,\,\,\, p=1,\cdots, P
\end{split}}
\end{equation}
where $\tilde{c}_{p,c}$ is the coarse motion parameter of $\tilde{c}_{p}$ corresponding to the RM correction and $\tilde{c}_{p,f}$ is the residual part of $c_p$ called fine motion parameter corresponding to the DFM correction.

According to (\ref{RM-DeltaR}), to avoid RM, the coarse motion parameter should be set as
\begin{equation}\label{Definition-of-coarse}
{\small\begin{split}
\tilde c_{p,c} &=\text{round}\left(\frac{c_{p}-c_{p,\min}}{\Delta c_{p,c}}\right) \Delta c_{p,c},+c_{p,\min}
\end{split}}
\end{equation}
with  $\Delta c_{p,c}$ satisfying $\Delta c_{p,c}= \Delta_{\text{RM},p}$. On the other hand, according to (\ref{RM-Deltafd}), to avoid DFM, the fine motion parameter should be set as
\begin{equation}
{\small\begin{split}
\label{Definition-of-fine} \tilde{c}_{p,f} =&\text{round}\left(\frac{c_{p}-\tilde{c}_{p,c}+\Delta_{p,c}/2}{\Delta c_{p,f}}\right) \Delta c_{p,f}-\frac{\Delta_{p,c}}{2}\\
\end{split}}
\end{equation}
with $\Delta c_{p,f}$ satisfying $\Delta c_{p,f}= \Delta_{\text{DRM},p}$.

Correspondingly,  the search spaces  of the coarse part of  target motion parameters are defined as
\begin{equation}
{\small\begin{split}\label{Space_c_c}
\mathbb{C}_{p,c}\triangleq [c_{p,\text{min}}:\Delta c_{p,c}:c_{p,\text{max}}], \,\,\,\,\, p=1,\cdots,P
\end{split}}
\end{equation}
while  the search spaces $\mathbb{C}_{p,f}$ of the fine part are defined as
\begin{equation}
{\small\begin{split}\label{Space_c_f}
\mathbb{C}_{p,f}\triangleq [-\Delta c_{p,c}/2:\Delta c_{p,f}:\Delta c_{p,c}/2], \,p=1,\cdots,P.
\end{split}}
\end{equation}
 The dual-scale search space is illustrated in Fig. \ref{Duale-scale}.
 \begin{figure}[h!]
\centering
\centerline{\includegraphics[width=8cm]{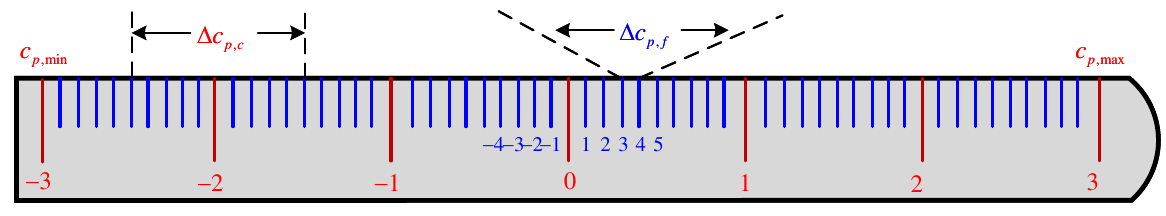}}
\caption{Dual-scale parameter decomposition.}
\label{Duale-scale}
\end{figure}

Notice that, in the case of the KT-MFP detector, the dual-scale decomposition of the velocity variable $\tilde{c}_{1}$ follows the same form of (\ref{velocity_Amb}), i.e.,
\begin{equation}\label{Partition-Parameter-c1}
{\small\begin{split}
\tilde{c}_{1}=\tilde{q}V_a+\tilde{c}'_{1}.
\end{split}}
\end{equation}
Correspondingly,  the search spaces of $\tilde{q}$ and $\tilde{c}'_{1}$ are defined as
\begin{equation}\label{Space_q}
\mathbb{Q}\triangleq [q_{\text{min}}:1:q_{\text{max}}]
\end{equation}
and
\begin{equation}\label{Space_baseband_velocity}
{\small\begin{split}
\mathbb{C}'_{1}\triangleq [-V_a/2:\Delta c_{p,f}:V_a/2].
\end{split}}
\end{equation}
 \begin{Rem}
Due to the increase of the step size, the size of the search space $\mathbb{C}_{p,c}$ of the coarse part is reduced to $f_s/f_c$ of  $\mathbb{C}_p$. Meanwhile, the step size of the fine part remains unchanged, but the search scope has been reduced to $\Delta c_{p,c}$ which is just the step size of the coarse parameter.
\end{Rem}
\subsection{Factorization of the GRFT detector}
Utilizing the dual-scale decomposition of motion parameters, we propose now a dual-scale implementation of the GRFT detector, which can dramatically decrease the computational complexity.


\begin{Pro}
Given the dual-scale decomposition in (\ref{Partition-Parameter}) - (\ref{Definition-of-fine}), the following holds:
\begin{equation}\label{GIFT-coarse}
{\small\begin{split}
&\emph{GIFT}_{y(f_k,t_m)}(\tilde{\mathbf{c}}_c)\\
\approx&A_4(\tau_n)\text{asinc}\left[B_r \left(\tau_{n}\!-\!\frac{2 c_0}{c}\right)\right]\\
&\times\exp\!\!\left[\!-j\frac{4\pi}{\lambda}  {\sum}_{p=1}^P (\tilde{c}_{p,c}+\kappa c_{p,f})  t_m^p \right]+\overline{w}'(\tau_n,t_m)\\
=&\emph{GIFT}_{y(f_k,t_m)}(\tilde{\mathbf{c}})\exp\left[j\frac{4\pi}{\lambda}{\sum}_{p=1}^P \frac{B_r}{2f_c} \, \tilde{c}_{p,f} t_m^p \right]
\end{split}}
\end{equation}
where:
 $A_4(\tau_n)\!\triangleq\!A'_2\! \, \exp\{j\pi B_r (\tau_n-2c_0/c)\}$;
 $\tilde{\mathbf{c}}_c\!\triangleq\![\tilde{c}_{0},\tilde{c}_{1,c},\cdots,\tilde{c}_{P,c}]$; $c_{p,f}\triangleq c_p- \tilde{c}_{p,c}$  for $p\!=\!1,\! \cdots\!,\! P$; $\kappa\triangleq 1-\frac{B_r}{2f_c}$;
  $\overline{w}'\!(\tau_n,t_m)\!\triangleq\!\emph{IFFT}\!\left\{\!w(f_k,t_m) \, \overline{H}_{\emph{RM}}(f_k,t_m; \tilde{\mathbf{c}}) \!\right\}\!H_\emph{pre-DFM}(t_m;\!\tilde{\mathbf{c}}_c)$ with $\overline{H}_{\emph{RM}}(\cdot)$ and $\overline H_\emph{pre-DFM}(\cdot)$ taking the forms (\ref{H-RM}) and (\ref{pre-DFM}) respectively.
\end{Pro}
\begin{proof}
Let $\overline{H}_{\text{RM}}(f_k,t_m;\tilde{\mathbf{c}}_c)$ denote the coarse compensation, i.e.,
\begin{equation}\label{H-RM-coarse}
{\small{\begin{split}
\overline{H}_{\text{RM}}(f_k,t_m;\tilde{\mathbf{c}}_c)&=\exp\left[j\frac{4\pi}{c} f_k \left({\sum}_{p=1}^P \tilde{c}_{p,c} t_m^p\right)\right].\\
\end{split}}}
\end{equation}
Then, substituting (\ref{H-RM-coarse}) into (\ref{Definition-GIFT}) yields
\begin{equation}\notag
{\small{\begin{split}
&\,\,\,\,\,\,\,\,\,\text{GIFT}_{y(f_k,t_m)}(\tilde{\mathbf{c}}_c)\\
&=\text{IFFT}\left\{y(f_k,t_m)\overline{H}_{\text{RM}}(f_k,t_m;\tilde{\mathbf{c}}_c)\right\}\\
&=\text{IFFT}\left\{\!A_1\exp\left[j \frac{4 \pi f_k}{c}  \left(c_0\!\!+{\sum}_{p=1}^{P} (c_p-\tilde{c}_{p,c})t_m^p\right)\right]\!\right\}\\
&\,\,\,\,\,\,\,\,\,\,\,\,\times\exp\left[j \frac{4 \pi }{\lambda}  \left({\sum}_{p=0}^{P} c_p t_m^p\right)\right]+\overline{w}'(\tau_n,t_m)\\
&=A'_2\text{asinc}\left[B_r \left(\tau_n\!-2\left(c_0\!+{\sum}_{p=1}^P c_{p,f} t_m^p\right)/c\right)\right]\\
&\,\,\,\,\,\,\,\,\,\times \exp\left[j\pi B_r \left(\tau_n\!-2(c_0+{\sum}_{p=1}^P (c_p-\tilde{c}_{p,c}) t_m^p)/c\right)\right]\\
&\,\,\,\,\,\,\,\,\,\times\exp\!\!\left[\!-j\frac{4\pi}{\lambda}\!  {\sum}_{p=1}^P c_{p} t_m^p \right]+\overline{w}'(\tau_n,t_m)\\
\end{split}}}
\end{equation}
\begin{equation}\label{GIFT-fine-GRFT}
{\small{\begin{split}
&=A_4(\tau_n)\text{asinc}\left[B_r \left(\tau_n\!-2\left(c_0\!+{\sum}_{p=1}^P c_{p,f} t_m^p\right)/c\right)\right]\\
&\times\exp\!\!\left[\!-j\frac{4\pi}{\lambda}\!\sum_{p=1}^P c_{p} t_m^p \right]\!\!\exp\left[j\pi B_r\!\!{\sum}_{p=1}^P 2(c_p-\tilde{c}_{p,c}) t_m^p/c\right]\\
&\,\,\,\,\,\,\,\,\,\,\,\,\,\,\,\,\,\,\,\,\,\,\,\,\,\,\,\,\,\,\,\,\,\,\,\,\,\,\,\,\,\,\,\,\,\,\,\,\,\,\,\,\,\,\,\,\,\,\,\,\,\,\,\,\,\,\,\,\,\,\,\,\,\,\,\,\,\,\,\,
\,\,\,\,\,\,\,\,\,\,\,\,\,\,\,\,\,\,\,\,\,\,\,\,\,\,\,\,+\overline{w}'(\tau_n,t_m).
\end{split}}}
\end{equation}

Clearly, according to the definition of fine compensations in (\ref{Definition-of-fine}), we can obtain the following inequality,
  \begin{equation}\label{Residual_RM_DeltaR}
 {\small\begin{split}
\left|  \sum_{p=1}^P c_{p,f} t_m^p\right|  & \leq \sum_{p=1}^P \left |c_{p,f}\right| T^p \\
&  \leq \sum_{p=1}^P \Delta c_{p,c} T^p\big/2\\
&\leq\Delta R\big/2.
\end{split}}
\end{equation}
Applying (\ref{Residual_RM_DeltaR}), the following approximation holds:
\begin{equation}\notag
{\small\begin{split}
\text{asinc}\left[B_r\!\left(\tau_n\!-\!2\left(c_0\!+\!\sum_{p=1}^P\!c_{p,f} t_m^p\right)\!/c\!\right)\right]\!\approx\!\text{asinc}\!\!\left[\!B_r(\!\tau_n\!-\!\frac{2 c_0}{c})\!\right].
\end{split}}
\end{equation}
Consequently, (\ref{GIFT-fine-GRFT}) can be further computed by
\begin{equation}\label{GRFT-GIFT-coarse-1}
{\small{\begin{split}
&\text{GIFT}_{y(f_k,t_m)}(\tilde{\mathbf{c}}_c)\\
\approx &A_4(\tau_n)\text{asinc}\!\!\left[\!B_r(\!\tau_n\!-\!\frac{2 c_0}{c})\!\right]\!\!\exp\!\!\left[\!-j\frac{4\pi}{\lambda}\!  {\sum}_{p=1}^P c_{p} t_m^p \right]\\
&\times\exp\left[j\pi B_r\!{\sum}_{p=1}^P 2(c_p-\tilde{c}_{p,c}) t_m^p/c\right]\!+\!\overline{w}'(\tau_n,t_m).\\
\end{split}}}
\end{equation}
Then, applying $c_{p,f}\triangleq c_{p}-\tilde{c}_{p,c}$ and $\kappa\triangleq 1-\frac{B_r}{2f_c}$, (\ref{GRFT-GIFT-coarse-1}) can be further simplified:
\begin{equation}\label{GRFT-GIFT-coarse-2}
{\small{\begin{split}
&\text{GIFT}_{y(f_k,t_m)}(\tilde{\mathbf{c}}_c)\\
= &\overline{w}'(\tau_n,t_m)+A_4(\tau_n)\text{asinc}\!\!\left[\!B_r(\!\tau_n\!-\!\frac{2 c_0}{c})\!\right]\\
&\,\,\,\,\,\,\,\,\,\,\,\,\,\,\,\,\,\,\,\,\,\,\,\,\,\,\,\,\,\,\,\,\,\,\times\exp\!\!\left[\!-j\frac{4\pi}{\lambda}\!  {\sum}_{p=1}^P (\tilde{c}_{p,c}+\kappa c_{p,f}) t_m^p \right].\\
\end{split}}}
\end{equation}

On the other hand, $\text{GIFT}_{y(f_k,t_m)}(\tilde{\mathbf{c}})$  can be written as:
\begin{equation}\label{GRFT-GIFT-complete}
{\small{\begin{split}
&\text{GIFT}_{y(f_k,t_m)}(\tilde{\mathbf{c}})\\
=&A_4(\tau_n)\text{asinc}\!\!\left[\!B_r(\!\tau_n\!-\!\frac{2 c_0}{c})\!\right]\!\!\exp\!\!\left[\!-j\frac{4\pi}{\lambda}\!  {\sum}_{p=1}^P c_{p} t_m^p \right]\\
&\,\,\,\,\,\times\exp\left[j\pi B_r\!\!{\sum}_{p=1}^P 2(c_p-\tilde{c}_{p}) t_m^p/c\right]+\overline{w}'(\tau_n,t_m).\\
\end{split}}}
\end{equation}
Comparing (\ref{GRFT-GIFT-coarse-1}) with (\ref{GRFT-GIFT-complete}), we have
\begin{equation}
{\small{\begin{split}
&\text{GIFT}_{y(f_k,t_m)}(\tilde{\mathbf{c}}_c)\\
\approx &\text{GIFT}_{y(f_k,t_m)}(\tilde{\mathbf{c}})\exp\left[j\frac{4\pi}{\lambda}{\sum}_{p=1}^P \frac{B_r}{2f_c}\tilde{c}_{p,f} t_m^p \right].
\end{split}}}
\end{equation}
\end{proof}

Proposition 1 has two implications:
\begin{itemize}
 \item After the coarse compensation, the motion parameters leading to the DFM can be factorized into the  expected coarse motion parameters $\tilde{c}_{p,c}$ and scaled residual ones $\kappa c_{p,f}=\left(1-\frac{B_r}{2f_c}\right) c_{p,f}$, $p=1,\cdots,P$.
  \item By the dual-scale decomposition (\ref{Partition-Parameter}),  the  RM effect can be eliminated by only  utilizing the coarse compensation coefficient $\overline{H}_{\text{RM}}(f_k,t_m;\tilde{\mathbf{c}}_c)$, instead of the complete compensation coefficient $\overline{H}_{\text{RM}}(f_k,t_m;\tilde{\mathbf{c}})$. Notice that compared with the complete compensation, the coarse compensation   induces another term of  phase $\exp\left[j\frac{4\pi}{\lambda}{\sum}_{p=1}^P \frac{B_r}{2f_c}\tilde{c}_{p,f} t_m^p \right]$ which  can be used to compensate the DFM caused by $-\frac{B_r}{2f_c} c_{p,f}$.
  \end{itemize}
\begin{Pro}\label{GRFT2Two-Step}
As long as $\frac{f_c/f_s}{M}\leq 1$ holds, we have
 \begin{equation}\label{DS-GRFT}
{\small{\begin{split}
\emph{GRFT}_{y(f_k,t_m)}(\tilde{\mathbf{c}})\approx \emph{GFT}_{\emph{mGIFT}_{y(f_k,t_m)}(\tilde{\mathbf{c}}_c)}(\tilde{\mathbf{c}}'_f)
\end{split}}}
\end{equation}
where: $\tilde{\mathbf{c}}=[\tilde{c}_{0}, \tilde{c}_{1},\cdots,\tilde{c}_{P}]$; $\tilde{\mathbf{c}}_c=[\tilde{c}_{0}, \tilde{c}_{1,c},\cdots,\tilde{c}_{P,c}]$; $\tilde{\mathbf{c}}_f=[\tilde{c}_{1,f},\cdots,\tilde{c}_{P,f}]$; $\tilde{\mathbf{c}}_f'=[\kappa\tilde{c}_{1,f},\cdots,\kappa\tilde{c}_{P,f}]$.
\end{Pro}
\begin{proof}
Substituting (\ref{Partition-Parameter}) into (\ref{standard-GRFT-FD}), we have
 \begin{equation}\label{TS-GRFT-step1}
{\small{\begin{split}
&\text{GRFT}_{y(f_k,t_m)}(\tilde{\mathbf{c}})\\
=&{\sum}_{m=1}^{M}\text{IFFT}\left\{y(f_k,t_m) \exp \left[j\frac{4\pi}{c}f_k {\sum}_{p=1}^P\tilde{c}_{p} t_m^p\right]\right\}\\
&\,\,\,\,\,\,\,\,\,\,\,\,\,\,\,\,\,\,\,\,\,\,\times \exp \left[j\frac{4\pi}{\lambda} \left({\sum}_{p=1}^P(\tilde{c}_{p,c}+\tilde{c}_{p,f}) t_m^p\right)\right]\\
=& \sum_{m=1}^{M}\text{GIFT}_{y(f_k,t_m)}(\tilde{\mathbf{c}})H_{\text{pre-DFM}}(t_m;\tilde{\mathbf{c}}_c)\\
&\,\,\,\,\,\,\,\,\,\,\,\,\,\,\,\,\,\,\,\,\,\,\,\,\,\,\,\,\,\,\,\,\,\,\,\,\,\,\,\,\,\,\,\times\overline{H}_{\text{DFM}}(t_m;\!\tilde{\mathbf{c}}_f)\exp[j\frac{4\pi}{\lambda}\tilde{c}_{1,f}t_m].\\
\end{split}}}
\end{equation}

Further, making use of Proposition 1, (\ref{TS-GRFT-step1}) can be represented as
\begin{equation}\label{TS-GRFT-step2}
{\small{\begin{split}
&\text{GRFT}_{y(f_k,t_m)}({\color{magenta}{\tilde{\mathbf{c}}}})\\
\!\approx &\sum_{m=1}^{M}\!\text{GIFT}_{y(f_k,t_m)}(\tilde{\mathbf{c}}_c)\exp\!\left[\!-j\frac{4\pi}{\lambda}{\sum}_{p=1}^P \frac{B_r}{2f_c}\tilde{c}_{p,f} t_m^p \right]\\
&\,\,\,\,\,\,\,\times H_{\text{pre-DFM}}(t_m;\tilde{\mathbf{c}}_c)\overline{H}_{\text{DFM}}(t_m;\!\tilde{\mathbf{c}}_f)\exp[j\frac{4\pi}{\lambda}\tilde{c}_{1,f}t_m]\\
\!=&\!\sum_{m=1}^{M}\!\text{mGIFT}_{y(f_k,t_m)}\!(\tilde{\mathbf{c}}_c)\!\exp\!\left[\!-j\frac{4\pi}{\lambda}{\sum}_{p=1}^P \frac{B_r}{2f_c}\tilde{c}_{p,f} t_m^p \right]\\
&\,\,\,\,\,\,\,\,\,\,\,\,\,\,\,\,\,\,\,\,\,\,\,\,\,\,\,\,\,\,\,\,\,\,\,\,\,\,\,\,\,\,\,\,\,\,\,\times\overline{H}_{\text{DFM}}(t_m;\!\tilde{\mathbf{c}}_f)\exp[j\frac{4\pi}{\lambda}\tilde{c}_{1,f}t_m]\\
\!=&\!\sum_{m=1}^{M}\!\text{mGIFT}_{y(f_k,t_m)}\!(\tilde{\mathbf{c}}_c)\!\overline{H}_{\text{DFM}}(t_m;\!\tilde{\mathbf{c}}_f')\!\exp[j\frac{4\pi}{\lambda}\tilde{c}_{1,f}'t_m] .
\end{split}}}
\end{equation}
 Eq. (\ref{TS-GRFT-step2})  shows conditioned on the coarse compensation procedures for RM  and DFM effects are completed, the residual fine compensation procedure for the DFM is independent of the previous procedure.

%
Notice that the RHS of (\ref{TS-GRFT-step2}) has the form of  $\sum_{m=1}^{M} f(x) \exp[j 2\pi\frac{2 x}{\lambda} t_m]$ which can be regarded as a DFT procedure. Only when $-V_{\text{a}}/2<\tilde{c}'_{1,f}<V_{\text{a}}/2$, this expression can be implemented by FFT efficiently. Thus, it is necessary to discuss  the relationship between $\tilde{c}'_{1,f}$ and $V_{\text{a}}$.

Clearly, we have $0<\kappa<1$. Then, according to (\ref{Definition-of-fine}) and (\ref{Space_c_f}), the fine part of the target velocity satisfies
 \begin{equation}\label{base-band-velocity}
{\small{\begin{split}
|\tilde{c}'_{1,f}| < |\tilde{c}_{1,f}|\leq\Delta c_{1,c}/2.
\end{split}}}
\end{equation}
Further, $\Delta c_{1,c}$ can be computed by
 \begin{equation}\label{delta-c1c}
{\small{\begin{split}
 \Delta c_{1,c}=\frac{c}{2 f_s T}=\frac{c \text{PRF}}{2 f_c}\frac{f_c/f_s}{M}=\frac{f_c/f_s}{M} V_{\text{a}} \, .
\end{split}}}
\end{equation}
Clearly, when the condition $\frac{f_c/f_s}{M} \leq 1$ holds, we have
\begin{equation}
{\small{\begin{split}
 \Delta c_{1,c} \leq V_{\text{a}}
\end{split}}}
\end{equation}
and thus
 \begin{equation}
{\small{\begin{split}\label{Range-of-Velocity}
| \tilde{c}'_{1,f}|<  \Delta c_{1,c}/2< V_{\text{a}}/2.
\end{split}}}
\end{equation}
Hence, under the condition $\frac{f_c/f_s}{M}\leq 1$, based on Definition 2, the GRFT procedure can be decomposed as
 \begin{equation}\label{GRFT2TwoStep-4}
{\small{\begin{split}
\text{GRFT}_{y(f_k,t_m)}(\tilde{\mathbf{c}})\approx\text{GFT}_{\text{mGIFT}_{y(f_k,t_m)}(\tilde{\mathbf{c}}_c)}(\tilde{\mathbf{c}}_f').
\end{split}}}
\end{equation}
\end{proof}

According to Proposition 2, by utilizing the dual-scale decomposition (\ref{Partition-Parameter}), the range-Doppler joint GRFT can be factorized into a GIFT process in the range domain and GFT processes in Doppler domain conditioned on the coarse motion parameter if $\frac{f_c/f_s}{M}\leq 1$.
Thanks to this appealing property, the joint correction of RM and DFM effects is decoupled into a cascade procedure, first RM correction on the coarse search space followed by DFM correction on the fine search spaces.  In this respect, this decoupled GRFT procedure is called \textit{Dual-Scale GRFT} (DS-GRFT).

We observe that the condition $\frac{f_c/f_s}{M}\leq 1$ can easily hold in practice. In fact:
\begin{itemize}
  \item  the number of pulses $M$ for LTCI is usually large, satisfying  $M\geq 256$ in practice;
  \item as far as  the ratio $f_c/f_s$ is concerned, for   high-resolution radars the bandwidth $B_r$ can range from hundreds of MHz to a few GHz ($f_s>B_r$ according to the
  Shannon-Nyquist sampling theorem), while for  low-resolution radars, which are usually used for long-range forewarning, the commonly used  carrier frequency $f_c$ is at L-band or S-band.
\end{itemize}
Overall, the condition $f_c/f_s\leq M$ can hold in most practical cases.

\begin{Rem}
According to (\ref{Partition-Parameter}), the fine part of target velocity $c_{1,f}$ falls into $[-\Delta c_{1,c}/2,\, \Delta c_{1,c}/2]$,  while the corresponding  fine velocity spectrum  estimated by $\emph{GRFT}_{y(f_k,t_m)}(\tilde{\mathbf{c}})$  is within $[-V_a/2,\, V_a/2]$. Thus, under the condition  $\frac{f_c/f_s}{M}\leq 1$, one needs to truncate $\emph{GRFT}_{y(f_k,t_m)}(\tilde{\mathbf{c}})$ beyond the domain $[-\Delta c_{1,c}/2,\, \Delta c_{1,c}/2]$.
For some applications which cannot satisfy the condition $\frac{f_c/f_s}{M} \leq 1$, the dual decomposition given in (\ref{TS-GRFT-step2}) still holds. However, since it cannot  be guaranteed  that the fine part of the velocity  satisfies (\ref{base-band-velocity}), the DFM correction cannot be efficiently implemented via FFT.
\end{Rem}

 According to \cite{GRFT}, the \textit{ML estimation} (MLE) of target motion parameters $\bf c$ is given by
\begin{equation}\label{MLE-GRFT-FD}
{\small{\begin{split}
[\hat{\bf c}_c,\hat{\bf c}_f']
=&\arg\max_{\tilde{\bf c}_c\in \mathbb {C}_c;\tilde{\bf c}_f'\in \kappa\mathbb {C}_f}  |\text{GFT}_{\text{mGIFT}_{y(f_k,t_m)}(\tilde{\mathbf{c}}_c)}(\tilde{\mathbf{c}}_f')|
\end{split}}}
\end{equation}
where: $\hat{\bf c}_c=[\hat c_0,\cdots, \hat c_{P,c}]$; $\hat{\bf c}_f'=[\hat c'_{1,f},\cdots,\hat c'_{P,f}]$; $\mathbb{C}_c=\mathbb{C}_0\times \cdots\times \mathbb{C}_{P,c}$ and $\mathbb{C}_f=\mathbb{C}_{1,f}\times \cdots\times \mathbb{C}_{P,f}$. Hence, the  target motion estimation is $\hat{\mathbf{c}}=[\hat c_0, \hat c_1,\cdots,\hat c_P]$ with $\hat{c}_p=\hat{c}_{p,c}+\hat{c}'_{p,f}/\kappa$ for $p\geq 1$. After the cascade procedure of RM and DFM compensations, a moving target can be finely focused in the fast-time range and slow-time Doppler domain, i.e.,
\begin{equation}\label{focused-RD}
{\small{\begin{split}
&\text{GFT}_{\text{mGIFT}_{y(f_k,t_m)}(\hat{\mathbf{c}}_c)}(\hat{\mathbf{c}}_f')\\
= &A_3\text{asinc}\left\{B \left(\tau_{\hat{n}}-\frac{2c_0}{c}\right)\right\} \\
&\times\text{asinc} \left\{\text{PRF}\left( \hat f_d(\hat c_{1,f}')-\frac{2 c_{1,f}'}{\lambda} \right) \right\}+\overline{w}'(\tau_n,t_m)\end{split}
}}
\end{equation}
where: $A_3$ denotes the  target complex attenuation after GIFT and GFT; $\hat{n}=\text{round}(2\hat{c}_0/cT_s)$ and $\hat f_d(\hat c_{1,f})$ is the Doppler frequency estimate, i.e., $\hat f_d(\hat c_{1,f})=\frac{2\hat c_{1,f}'}{\lambda}$.
\subsection{Factorization of  the  KT-MFP detector}
The decomposition (\ref{velocity_Amb}) can be  considered as another dual-scale composition into  baseband velocity $c'_1$ and  ambiguity velocity $qV_a$.  Thanks to the KT, the RM caused by the baseband velocity is first eliminated. Consequently, the  first order RM and Doppler frequency shift are caused by $q V_a$ and $c'_1$, respectively.
In this respect, corrections of first order RM and Doppler frequency shift are  performed  on dual-scale spaces. However,   since the higher-order RM and DFM  (say for $p\geq2$) are  related to the same motion parameters,  the corresponding correction for standard KT-MFP is still performed on the unified single-scaled space. To this end, we  further study  factorization of the KT-MFP detector utilizing the dual-scale  decomposition in (\ref{Partition-Parameter}).
\begin{Cor}
Given the dual-scale decomposition in (\ref{Partition-Parameter}) - (\ref{Definition-of-fine}), the following holds:
\begin{equation}\label{GIFT-KT-coarse}
{\small{\begin{split}
&\emph{GIFT}_{y_{\emph{KT}}(f_k,t_m)}(\tilde{\mathbf{c}}_c)\\
\approx&A_4(\tau_n)\emph{asinc}[B_r (\tau_n-2 c_0/c)]\exp\!\!\left[\!-j\frac{4\pi}{\lambda} c_{1} t_m \right]\\
&\,\,\,\,\,\,\,\,\,\,\,\,\,\,\,\,\,\times\exp\left[\!-j\frac{4\pi}{\lambda}(c_{2,c}+\kappa c_{2,f}) t_m^2 \right]+\overline{w}'(\tau_n,t_m)\\
=&\emph{GIFT}_{y_{\text{KT}}(f_k,t_m)}(\tilde{\mathbf{c}})\exp\left[-j\frac{4\pi}{\lambda}\frac{B_r}{2f_c}\tilde{c}_{2,f} t_m^2 \right]\end{split}}}
\end{equation}
where:
$\tilde{\mathbf{c}}_c\triangleq[\tilde{c}_{0},\tilde{q},\tilde{c}_{2,c}]$; $c_{p,f}\!\triangleq\!c_{p}-\tilde{c}_{p,c}$ for $ p=1, 2$; $\overline{w}'(\tau_n,t_m)\!\triangleq\!\text{IFFT}\left\{\overline{w}(f_k,t_m)\overline{H}_{\emph{KT},\emph{RM}}(f_k,t_m; \tilde{\mathbf{c}}_c) \right\} \overline H_\emph{pre-DFM}(t_m;\tilde{\mathbf{c}}_c)$, the compensation function $\overline{H}_{\emph{KT},\emph{RM}}(\cdot)$ taking the form (\ref{H-MFP-RM}) and $\overline H_\emph{pre-DFM}(\cdot)$ the form (\ref{pre-DFM}).
\end{Cor}
\begin{proof}
Substituting expressions of the coarse compensation $\overline{H}_{\text{KT},\text{RM}}(f_k,t_m;\tilde{\mathbf{c}}_c)$, $H_{\text{pre-DFM}}(t_m;\tilde{\mathbf{c}}_c)$ and $y_\text{KT}(f_k,t_m)$ into (\ref{Definition-GIFT}) yields
\begin{equation}\label{GIFT-KT-fine}
{\small{\begin{split}
&\text{GIFT}_{y_{\text{KT}}(f_k,t_m)}(\tilde{\mathbf{c}}_c)\\
=&\text{IFFT}\!\left\{y_{\text{KT}}(f_k,t_m)\overline{H}_{\text{KT},\text{RM}}(f_k,t_m;\tilde{\mathbf{c}}_c)\right\}\\
=&A_4(\tau_n)\text{asinc}\left[B_r \left(\tau_n-2(c_0+ c_{2,f} t_m^2)/c\right)\right]\\
&\,\,\,\,\,\,\,\,\,\,\,\,\,\,\,\,\,\,\,\,\,\,\,\,\,\,\times\exp\left\{j\pi B_r 2(c_2-\tilde{c}_{2,c}) t_m^2/c\right\}\\
&\,\,\,\,\,\,\,\,\,\,\,\,\,\,\,\,\,\,\,\,\,\,\,\,\times\exp\!\!\left[\!-j\frac{4\pi}{\lambda}\!  {\sum}_{p=1}^2 c_{p} t_m^p \right]+\overline{w}'(\tau_n,t_m)\, .
\end{split}}}
\end{equation}
Similarly, we have $|c_{2,f} t_m^2|\leq|\Delta c_{2,c} T^2/2|\leq \Delta R\big/2$, then (\ref{GIFT-KT-fine}) can be further expressed as
\begin{equation}\label{GIFT-KT-coarse-2}
{\small{\begin{split}
&\text{GIFT}_{y_{\text{KT}}(\!f_k,t_m\!)}(\tilde{\mathbf{c}}_c)\\
\approx&A_4(\tau_n)\text{asinc}[B_r (\tau_n\!-\!2 c_0/c)]\!\exp\!\!\left[\!-j\frac{4\pi}{\lambda}\!  {\sum}_{p=1}^2 c_{p} t_m^p \right]\\
&\,\,\,\,\,\,\,\,\,\,\,\,\,\times\exp\left\{j\pi B_r 2(c_2-\tilde{c}_{2,c}) t_m^2/c\right\}+\overline{w}'(\tau_n,t_m)\\
=&A_4(\tau_n)\text{asinc}[B_r (\tau_n-2 c_0/c)]\exp\!\!\left[\!-j\frac{4\pi}{\lambda} c_{1} t_m \right]\\
&\,\,\,\,\,\,\,\,\,\,\,\,\times\exp\left[\!-j\frac{4\pi}{\lambda}(c_{2,c}+\kappa c_{2,f}) t_m^2 \right]+\overline{w}'(\tau_n,t_m).
\end{split}}}
\end{equation}

On the other hand, $\text{GIFT}_{y_{\text{KT}}(f_k,t_m)}(\tilde{\mathbf{c}})$ can be rewritten as:
\begin{equation}
{\small{\begin{split}
&\text{GIFT}_{y_{\text{KT}}(f_k,t_m)}(\tilde{\mathbf{c}})\\
=&A_4(\tau_n)\text{asinc}\!\!\left[\!B_r(\!\tau_n\!-\!\frac{2 c_0}{c})\!\right]\!\!\exp\!\!\left[\!-j\frac{4\pi}{\lambda}\!  {\sum}_{p=1}^2 c_{p} t_m^p \right]\\
&\,\,\,\,\,\,\,\,\,\,\,\,\,\times\exp\left[j\pi B_r 2(c_2-\tilde{c}_{2}) t_m^2/c\right]+\overline{w}'(\tau_n,t_m)\\
\approx&\text{GIFT}_{y_{\text{KT}}(f_k,t_m)}(\tilde{\mathbf{c}}_c)\exp\left[-j\frac{4\pi}{\lambda}\frac{B_r}{2f_c}\tilde{c}_{2,f} t_m^2 \right].
\end{split}}}
\end{equation}
\end{proof}
\begin{Pro}\label{KT-GFT2Two-Step}
Given the echo $y_{\emph{KT}}(f_k,t_m)$ of the form (\ref{Echo-after-KT}) in range-frequency and slow-time domain after KT, the following holds:
 \begin{equation}\label{DS-KT-MFP}
{\small{\begin{split}
\emph{KT-MFP}_{y_{\emph{KT}}(f_k,t_m)}(\tilde{\mathbf{c}})
\approx \emph{GFT}_{\emph{mGIFT}_{y_{\emph{KT}}(f_k,t_m)}(\tilde{\mathbf{c}}_c)}(\tilde{\mathbf{c}}_f')
\end{split}}}
\end{equation}
where: $\tilde{\mathbf{c}}=[\tilde{c}_{0},\tilde{c}_{1},\tilde{c}_{2}]$; $\tilde{\mathbf{c}}_c=[\tilde{c}_{0},\tilde{q},\tilde{c}_{2,c}]$; $\tilde{\mathbf{c}}_f=[\tilde{c}'_{1}, \tilde{c}_{2,f}]$; $\tilde{\mathbf{c}}_f'=[\tilde{c}'_{1},\kappa\tilde{c}_{2,f}]$.
\end{Pro}
\begin{proof}
Substituting (\ref{Partition-Parameter}) into (\ref{KT-GRFT}) yields
 \begin{equation}\label{KT-GRFT2TwoScal-step1}
{\small{\begin{split}
&\text{KT-MFP}_{y_{\text{KT}}(f_k,t_m)}(\tilde{\mathbf{c}})\\
=&\sum_{m=1}^{M}\!\text{IFFT}\!\!\left\{\!y_{\text{KT}}(f_k,t_m)\!\exp\!\!\left[\!j\frac{4\pi}{c}f_k(\overline{f}\tilde{q} V_a t_m\!-\!\tilde{c}_{2} t_m^2)\right]\!\right\}\\
&\,\,\,\,\,\,\,\,\,\,\,\,\,\,\,\,\,\,\,\,\times\exp \left[j\frac{4\pi}{\lambda} (\tilde{c}_{2,c}\!+\!\tilde{c}_{2,f}) t_m^2\right]\!\exp\left[\!j\frac{4\pi}{\lambda}\tilde{c}'_1t_m\right]\\
=&\sum_{m=1}^{M}\text{GIFT}_{y_{\text{KT}}(f_k,t_m)}(\tilde{\mathbf{c}})H_{\text{pre-DFM}}(t_m;\tilde{c}_{2,c})\\
&\,\,\,\,\,\,\,\,\,\,\,\,\,\,\,\,\,\,\,\,\,\,\,\,\,\,\,\,\,\,\,\,\,\,\,\,\,\,\,\,\times\overline{H}_{\text{DFM}}(t_m;\tilde{c}_{2,f})\!\exp\left[\!j\frac{4\pi}{\lambda}\tilde{c}'_1t_m\right]
\end{split}}}
\end{equation}
Further, applying \textbf{Corollary 1}, (\ref{KT-GRFT2TwoScal-step1})  can be further  rewritten as
 \begin{equation}\label{KT-GRFT2TwoScal-step2}
{\small{\begin{split}
&\text{KT-MFP}_{y_{\text{KT}}(f_k,t_m)}(\tilde{\mathbf{c}})\\
\!\!\!&\approx\sum_{m=1}^{M}\! \text{GIFT}_{y_{\text{KT}}(f_k,t_m)}\!(\tilde{\mathbf{c}}_c)\exp\left[-j\frac{4\pi}{\lambda}\frac{B_r}{2f_c}\tilde{c}_{2,f} t_m^2 \right]\\
&\,\,\,\times H_{\text{pre-DFM}}(t_m;\tilde{c}_{2,c})\overline{H}_{\text{DFM}}(t_m;\tilde{c}_{2,f}\!)\exp \!\left[ j\frac{4\pi}{\lambda} \tilde{c}_1' t_m \!\right]\\
\!\!\!=&\!\sum_{m=1}^{M}\!\!\text{mGIFT}_{y_{\text{KT}}(f_k,t_m)}\!(\tilde{\mathbf{c}}_c\!)\overline{H}_{\text{DFM}}(t_m;\!\tilde{c}_{2,f}'\!) \!\exp \!\!\left[j\!\frac{4\pi}{\lambda} \tilde{c}_1' t_m \!\right].
\end{split}}}
\end{equation}
Since $|c_1'|\leq V_{a}/2$,  (\ref{KT-GRFT2TwoScal-step2}) can be efficiently implemented using the GFT according to Definition 2, i.e.,
\begin{equation}
\begin{split}
\text{KT-MFP}_{y_{\text{KT}}(f_k,t_m)}(\tilde{\mathbf{c}})\approx\text{GFT}_{\text{mGIFT}_{y_{\text{KT}}(f_k,t_m)}(\tilde{\mathbf{c}}_c)}(\tilde{\mathbf{c}}_f').
\end{split}
\end{equation}
\end{proof}
According to Proposition 3,  the KT-MFP procedure can be also factorized into a GIFT process in the range domain on the coarse search space and  a  GFT process in Doppler domain conditioned on the coarse motion parameter. Then, the joint correction of RM and DFM effects is decoupled into a cascade procedure, first RM correction on the coarse search space followed by  DFM correction on the fine search spaces.  This decoupled KT-MFP procedure is called \textit{Dual-Scale KT-MFP} (DS-KT-MFP).

For the DS-KT-MFP,  the MLE of target motion parameters $\bf c$ is given by
\begin{equation}\label{LRT-KT-MFP}
{\small{\begin{split}
[\hat{\bf c}_c,\hat{\bf c}_f']
\!=\!\arg\max_{\tilde{\bf c}_c\in \mathbb {C}_c;\tilde{\bf c}'_f\in \kappa\mathbb {C}_f}  |\text{GFT}_{\text{mGIFT}_{y_\text{KT}(f_k,t_m)}(\tilde{\mathbf{c}}_c)}(\tilde{\mathbf{c}}_f')|
\end{split}}}
\end{equation}
where: $\hat{\bf c}_c=[\hat c_0,\hat q,\hat c_{2,c}]$; $\hat{\bf c}_f'=[\hat c'_1,\hat c_{2,f}']$; $\mathbb{C}_c=\mathbb{C}_0\times \mathbb{Q}\times\mathbb{C}_{2,c}$ and $\mathbb{C}_f=\mathbb{C}'_{1}\times \mathbb{C}_{2,f}$. Hence, the  target motion estimation is $\hat{\mathbf{c}}=[\hat c_0, \hat q V_a+\hat c'_1, \hat c_{2,c}+\hat{\mathbf{c}}'_f/\kappa]$.
\begin{Rem}
As to the case of multiple targets, the signal model in (\ref{Rx-OFDM-Fre1}) should be re-expressed as follows:
\begin{equation}\label{Rx-OFDM-Multi-Target}
{\small{\begin{split}
\!\!\!y(f_k,t_m)\!\!=\!\!\sum_{\ell=1}^L\! \!A_1^{(\ell)} \emph{exp} \left[ -j\frac{4\pi}{c}\!(f_k\!+\!f_c) R^{(\ell)}(t_m) \right]\!+\!w(f_k,t_m)
\end{split}}}
\end{equation}
where $L$ denotes the number of targets.
Then, the DS-GRFT and  DS-KT-MFP  detectors have the same form of  (\ref{DS-GRFT}) and (\ref{DS-KT-MFP}), respectively, by replacing the input signal by (\ref{Rx-OFDM-Multi-Target}).
\end{Rem}

\subsection{Advantages of DS-GRFT family}
Compared to the standard GRFT family, the DS-GRFT family yields the following advantages.
\begin{itemize}
\item
Thanks to the proposed dual-scale decomposition, the RM correction  of the DS-GRFT family is more computationally efficient than that of standard ones, since  the number of  IFFT operations used for RM correction  is greatly reduced. Specifically, for the DS-GRFT, such a number is reduced  of a factor  $(f_c/f_s)^2$  with respect to the standard one when $P=2$.

\item Moreover, a range spectrum with no RM effect is also obtained in advance, thus giving more freedom
for the subsequent DFM compensation. Based on this range spectrum, one can correct the
DFM  within   the scope of range bins of interest, rather than within the whole range scope, which can also help to save execution time.

\item
Again, due to the dual-scale decomposition, the coarse motion parameters for the DFM correction can be pre-compensated at an early stage. Then, for  the residual fine motion parameter compensation,  the search space is tailored into a much smaller scope $[-V_a/2,V_a/2$]. As such,  the DFM correction can be implemented using  efficient FFT operations rather than MF operations, greatly reducing the computational burden.
\end{itemize}

Besides reducing the computational complexity, the family of DS-GRFT detectors also provides implementation advantages with respect to the single-scale space implementation.  The decoupling of the joint correction greatly improves the freedom of algorithm implementation, e.g., one can adopt different algorithms, in place of GIFT or GFT, for the RM and  DFM corrections,  considering practical requirements such as  detection performance, computational efficiency, etc.

\section{Implementation Issues }

\subsection{Pseudo-code of the proposed DS-GRFT detector family}
This subsection provides pseudo-codes of the proposed family of DS-GRFT detectors. According to Propositions 2 and 3, both DS-GRFT and DS-KT-MFP involve two cascade procedures, i.e.,  mGIFT compensation with respect to the coarse motion parameters followed by GFT compensation with respect to the fine motion parameters.
Hence, we first provide pseudo-codes of  mGIFT and GFT operations, in Algorithms 1 and 2, respectively.  Then, pseudo-codes of the DS-GRFT and DS-KT-MFP (taking, for instance, the order of motion $P=2$) are given in Algorithms 3 and 4, respectively.
Note that Algorithm 3 can be easily extended to DS-GRFT with $P>2$. As for the DS-KT-MFP with $P=3$, it also can be easily extended by the combination of \cite{KT-MFP-Three-Order} and Algorithm 4.
\begin{algorithm}[htb]\label{algorithm: LM-GCI-LMB}
\caption{$\text{mGIFT}$}
\small
{\underline{INPUT:}  Baseband signal   $y(f_k,t_m)$ of form (\ref{Rx-OFDM-Fre2}) or (\ref{Echo-after-KT}); coarse motion parameters $\mathbf{c}_c$\;
\underline{OUTPUT:} Signal $y(\tau,t_m;\mathbf{c}_{c})$ after coarse compensation\;
\textbf{function}  $\text{mGIFT}(y(\tau,t_m),\mathbf{c}_{c})$\\
Construct the coarse compensation  $H_{\text{RM}}(f_k,t_m;\mathbf{c}_{c})$  via (\ref{H-RM}) or (\ref{H-MFP-RM})\;
Construct the pre-compensation   $H_{\text{pre-DFM}}(t_m;\mathbf{c}_{c})$  via (\ref{pre-DFM})\;
Multiply $y(f_k,t_m)$ by $H_{\text{RM}}(f_k,t_m;\mathbf{c}_{c})$ and $H_{\text{pre-DFM}}(t_m;\mathbf{c}_{c})$\;
Perform the IFFT in the Range frequency domain for $m=1,\cdots,M$, i.e.,
$$y(\tau,t_m;\mathbf{c}_{c})=\text{IFFT}\{y(f_k,t_m)H_{\text{RM}}(f_k,t_m;\mathbf{c}_{c})\}H_{\text{pre-DFM}}(t_m;\mathbf{c}_{c})$$
\textbf{return}: $y(\tau,t_m;\mathbf{c}_{c})$.}
\end{algorithm}

\begin{algorithm}[htb]\label{algorithm: LM-GCI-LMB}
\caption{$\text{GFT}$}
\small
{\underline{INPUT:}  Signal   $y(\tau,t_m;\mathbf{c}_c)$ output by $\text{mGIFT}$; fine motion parameters $\mathbf{c}_{f}$\;
\underline{OUTPUT:} Signal $y(\tau,f_d;\mathbf{c}_f|\mathbf{c}_c)$ after fine compensation\;
\textbf{function}  $\text{GFT}(y(\tau,t_m;\mathbf{c}_{c}),\mathbf{c}_f)$\\
Construct the fine compensation $H_{\text{DFM}}(t_m;\mathbf{c}_f)$  via (\ref{H-DFM})\;
Multiply $H_{\text{DFM}}(t_m;\mathbf{c}_f)$  by $y(\tau,t_m;\mathbf{c}_c)$\;
Perform FFT in the Doppler domain for each range bin of interest, i.e.,
$$y(\tau,v_d;\mathbf{c}_f|\mathbf{c}_c)=\text{FFT}\{y(\tau,t_m;\mathbf{c}_c)H_{\text{DFM}}(t_m;\mathbf{c}_f)\}$$
\textbf{return}: $y(\tau,v_d;\mathbf{c}_f|\mathbf{c}_c)$.}
\end{algorithm}
\begin{algorithm}[h!]\label{algorithm: LM-GCI-LMB}
\caption{DS-GRFT}
\small
{\underline{INPUT:}  Baseband signal   $y(f_k,t_m)$ of the form (\ref{Rx-OFDM-Multi-Target}) \;

\underline{OUTPUT:} motion parameter estimates $\hat{\mathbf{c}}$\;
Initialize $\hat {\bf c}=\emptyset$\;
\textbf{function} DS-GRFT \,($y(f_k,t_m)$)\\

\For{$(\tilde{c}_{1,c},\tilde{c}_{2,c})\in\mathbb{C}_{1,c}\times\mathbb{C}_{2,c}$}
{
$\tilde{\mathbf{c}}_c:=[\tilde{c}_{1,c},\tilde{c}_{2,c}]$\;
$
y(\tau_n,t_m;\tilde{\mathbf{c}}_c):=\text{mGIFT}(y(f_k,t_m),\tilde{\mathbf{c}}_c)$\;
\For{$\tilde{c}_{2,f}'\in\kappa\mathbb{C}_{2,f}$}
{
$\tilde{\mathbf{c}}_f':=\tilde{c}_{2,f}'$\;
$y(\tau_n,v_d;\tilde{\mathbf{c}}_f'|\tilde{\mathbf{c}}_c):=\text{GFT}(y(\tau_n,t_m;\tilde{\mathbf{c}}_c),\tilde{\mathbf{c}}_f')$\;
\For{$(\tau_n,v_d)\in2\mathbb{C}_0/c\times\kappa\mathbb{C}_{1,f}$}
{
\If{$|y(\tau_n,v_d;\tilde{\mathbf{c}}_f'|\tilde{\mathbf{c}}_c)|\!>\!\gamma_{\emph{LRT}}$}
{
$\hat{c}_{0}:=c\tau_n/2$; $\hat{c}_{1}:=\tilde{c}_{1,c}+v_d/\kappa$; $\hat{c}_{2}:=\tilde{c}_{2,c}+\tilde{c}_{2,f}/\kappa$\; $\hat{\mathbf{c}}:=\hat{\mathbf{c}}\bigcup\{[\hat{c}_{0},\hat{c}_{1},\hat{c}_2]\}$\;
}}}}
\textbf{return}: $\hat{\mathbf{c}}$.}
\end{algorithm}

\begin{algorithm}[h!]\label{algorithm: LM-GCI-LMB}
\caption{DS-KT-MFP}
\small
{\underline{INPUT:}  Baseband signal   $y(f_k,t_m)$ of the form (\ref{Rx-OFDM-Multi-Target}) \;

\underline{OUTPUT:}  motion parameter estimates $\hat{\mathbf{c}}$;

\textbf{function} DS-KT-MFP \,($y(f_k,t_m)$)\;
Perform KT on $y(f_k,t_m)$, yielding $y_{\text{KT}}(f_k,t_m)$ of the form (\ref{Echo-after-KT})\;
\For{$(\tilde{q},\tilde{c}_{2,c})\in \mathbb{Q}\times \mathbb{C}_{2,c}$}
{
$\tilde{\mathbf{c}}_c:=[\tilde{q},\tilde{c}_{2,c}]$\;
$
y(\tau_n,t_m;\tilde{\mathbf{c}}_c):=\text{mGIFT}(y_{\text{KT}}(f_k,t_m),\tilde{\mathbf{c}}_c)$\;
\For{$\tilde{c}_{2,f}'\in \kappa\mathbb{C}_{2,f}$}
{
$\tilde{\mathbf{c}}_f':=[\tilde{c}_{2,f}']$\;
$y(\tau_n,v_d;\tilde{\mathbf{c}}_{f}'|\tilde{\mathbf{c}}_c):=\text{GFT}(y_{\text{KT}}(\tau_n,t_m;\tilde{\mathbf{c}}_c),\tilde{\mathbf{c}}_f')$\;
\For{$(\tau_n, v_d)\in 2\mathbb{C}_0/c \times \mathbb{C}_1'$}
{
\If{$|y(\tau_n,v_d;\tilde{\mathbf{c}}_f'|\tilde{\mathbf{c}}_c)|\!>\!\gamma_{\emph{LRT}}$}
{
$\hat{c}_{0}=c\tau_n/2$; $\hat{c}_{1}=\tilde{q}V_a+v_d$; $\hat{c}_{2}=\tilde{c}_{2,c}+\tilde{c}_{2,f}/\kappa$\;  $\hat{\mathbf{c}}:=\hat{\mathbf{c}}\bigcup\{[\hat{c}_{0},\hat{c}_{1},\hat{c}_2]\}$\;
}}}}
\textbf{return}: $\hat{\mathbf{c}}$. }
\end{algorithm}
\subsection{Computational complexity analysis}
\begin{table*}[hbt]
\renewcommand{\arraystretch}{2}
\caption{Computational complexity of GRFT detectors}\label{Computation_Complexity_Comparison_of Searching}
\begin{center}
\footnotesize
\scalebox{0.85}{ \begin{tabular*}{0.95\textwidth}{@{\extracolsep{\fill}}c| c c c c ||c}
\hline
\hline
 \diagbox{Alg.s}{\rotatebox{-35}{Times}}{\rotatebox{-70}{BOs}}& \makecell{\text{KT}  \\ {\footnotesize{$\mathcal{O}(M^2N_0)$}}}   &  \makecell{MF\\ {\footnotesize{$\mathcal{O}(M)$}}} &  \makecell{\text{IFFT} \\ {\footnotesize{$\mathcal{O}(0.5N_0 \log_2N_0)$}}}                                  &   \makecell{\text{FFT}\\ {\footnotesize{$\mathcal{O}(0.5M\log_2M)$}}} & \makecell{\text{Total Computational Complexity }} \\
\hline
TD-GRFT                  &  --             &  $N_{2} N_{1}N_0 $    & --                                            &          --       &   $\mathcal{O}(MN_0 N_{1}N_{2}) $           \\
\hline
FD-GRFT                     &  --             &  $N_{2} N_{1}N_0 $    & $N_{2} N_{1}  $                          &         --        &    $\mathcal{O}(MN_0 N_{1}N_{2} +0.5N_{2} N_{1}N_0 \log_2N_0)$        \\
\hline
DS-GRFT                  &  --             &  $\frac{f_s}{f_c} N_{2} N_{1}\widetilde{N}_0 $   & $(\frac{f_s}{f_c} )^2 N_{2} N_{1}M $     & $\frac{f_s}{f_c} N_{2} N_{1}\widetilde{N}_0$ & \makecell{ $\mathcal{O}(\frac{f_s}{f_c}M\widetilde{N}_0N_{2} N_{1} +\frac{f_s^2}{2f_c^2}  N_{2} N_{1}MN_0 \log_2N_0$\\$+ \frac{f_s}{2f_c} N_{2} N_{1}\widetilde{N}_0M\log_2M)$}\\
\hline
KT-MFP                   & $1$             & $ N_{2} N_{a}N_0$                  &$N_{2} N_{a}M$                & $N_{2} N_{a}\widetilde{N}_0 $    & \makecell{ $\mathcal{O}(N_{2}N_{a}MN_0 + 0.5N_{2} N_{a}MN_0 \log_2N_0$ \\$+M^2N_0+ 0.5 N_{2} N_{a}\widetilde{N}_0M\log_2M)$ }                        \\
\hline
DS-KT-MFP                &  $1$            & $ N_{2} N_{a}\widetilde{N}_0$                   & $\frac{f_s}{f_c}  N_{2} N_{a} M$           & $N_{2} N_{a} \widetilde{N}_0$  & \makecell{ $\mathcal{O}(N_{2}N_{a}M\widetilde{N}_0 + \frac{f_s}{2f_c}N_{2} N_{a}MN_0 \log_2N_0$ \\$+ M^2N_0+ 0.5N_{2} N_{a}\widetilde{N}_0M\log_2M)$ }\\
\hline
\hline
\end{tabular*}}
\label{computing-complex}
\normalsize
\end{center}
\end{table*}
In what follows, the computational complexity of the DS-GRFT family (including  DS-GRFT and DS-KT-MFP detectors) is analyzed by comparison with the TD-GRFT detector, the standard GRFT family (including FD-GRFT and KT-MFP detectors). For ease of exposition, the case of target motion order $P=2$  will be considered.  Let $M$ denote the number of pulses, $N_0$  the number of range bins, $K$  the number of frequency bins (i.e., $K=N_0$), $N_{1}$ the number of velocity bins, $N_{2}$  the number of acceleration bins, $N_a$ the number of ambiguity velocities, and $\widetilde{N}_0$ the number of range bins for the area of interest.

According to  implementations of the five considered detectors, their computational costs are mainly due to
four basic operations (BOs): MF, FFT, IFFT and KT. Then computational complexity of  these four basic operations and overall complexity for the different detectors are provided in Table \ref{Computation_Complexity_Comparison_of Searching}.
{\color{magenta}As it can be seen, compared to the FD-GRFT detector,  the TD-GRFT detector has  lower complexity at the expense of worse detection performance consistently with the analysis given in Section III-A.}
As for the  computational complexity comparison between standard GRFT and DS-GRFT families,  the following  conclusions can be drawn.

 1) The most computationally  demanding part of the FD-GRFT is the MF operation, while for the DS-GRFT is the FFT operation.
 Thus, the computational complexity of the latter is reduced to about $\frac{f_s \widetilde{N}_0}{2f_c N_0}\log_2(M)$ of the former. Since $f_c$ is usually much bigger than $f_s$, the computational efficiency improvement achieved by  DS-GRFT  is generally significant.

 2) For KT-MFP and DS-KT-MFP detectors, the most computationally  burdensome tasks  are IFFT  and FFT operations, respectively. Through a rough calculation, it can be seen that the computational complexity of the latter is decreased to about $\frac{\widetilde{N}_0\log_2(M)}{N_0\log_2(N_0)}$ of the former. Usually, range bins for the area of interest  are corresponding to the  far field of the radar, where  targets are usually difficult to be detected due to the low
 \textit{signal-to-noise-ratio} (SNR). Thus  $\widetilde{N}_0$ could be half or a quarter of $N_0$ (or even lower than a quarter) generally. With this respect, the computational efficiency improvement of DS-KT-MFP is significant.

 In addition, compared to the DS-GRFT detector, the complexity of the DS-KT-MFP detector is only $\frac{f_c N_a}{f_s N_{1}}$ of the former. Actually, the factor $\frac{N_{c1}}{f_c/f_s}$ equals the number of coarse velocity bins.
 In particular, when $\frac{f_c}{f_s}=\frac{N_{1}}{N_a}$, the proposed two dual-scale detectors (i.e., DS-GRFT and DS-KT-MFP)   achieve similar  computational efficiency.
\section{Performance assessment}
{\color{magenta}{In this section, coherent integration performance  and  computational efficiency of the proposed detectors are assessed  in the presence of Gaussian noise through two synthetic experiments}}. In the considered scenario (see Fig. \ref{scenario-UVA-detection}) an \textit{Integrated Communication Sensing Radar} (ICSR)  \cite{ICSR} is deployed, working in the millimeter wave frequency band and  surveilling \textit{Unmanned Aerial Vehicles} (UAVs).
The system parameters of the radar are summarized in Table \ref{radar-para.}.
Generally, the \textit{Radar Cross-Section} (RCS) of an UAV is  small and the transmission power of the ICSR  is restricted by the NG-RAN protocol.
With this respect, an LTCI algorithm is needed for the ICSR in order to detect weak targets.
\begin{figure}[h!]
\centering
\centerline{\includegraphics[width=8cm]{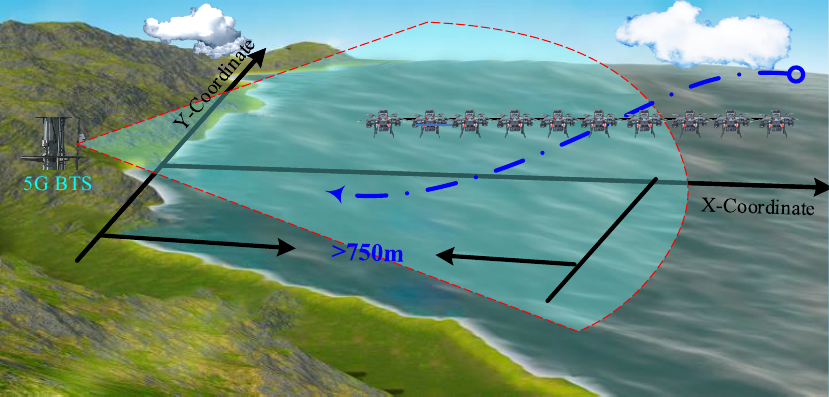}}
\caption{Scenario of UAV surveillance via ICSR system.}
\label{scenario-UVA-detection}
\end{figure}

\begin{table}[th]
\renewcommand{\arraystretch}{1.5}
\caption{Radar parameters}\label{radar-para.}
\begin{center}
\centering
\footnotesize
\begin{tabular*}{0.48\textwidth}{@{\extracolsep{\fill}}c c c c  c}
\hline
\hline
$f_c $ [GHz]            & $f_s$ [MHz]    &   $B_r$ [MHz]    &   PRF [Hz]    & M\\
\hline
 28                        &491.52          & 400           &   1905             & 512\\
\hline
\hline
\end{tabular*}
\normalsize
\end{center}
\end{table}

\begin{table}[h!]
\renewcommand{\arraystretch}{1.5}
\caption{Target motion parameters in Experiment 1}\label{target-para}
\begin{center}
\centering
\footnotesize
\begin{tabular*}{0.45\textwidth}{@{\extracolsep{\fill}}c| c|c|c}
\hline
\hline
Num.           &   slant range [$\text{m}$]  &     velocity [$\text{m}/\text{s}$]  &   acceleration   [$\text{m}/\text{s}^2$]      \\
\hline
target 1       &    538.52             &     20          &   5.07          \\
target 2       &    538.52             &     21          &    5.07          \\
target 3       &    492.44             &     17          &    7.58          \\
\hline
\hline
\end{tabular*}
\normalsize
\end{center}
\end{table}
\subsection{Experiment 1}
The purpose of this experiment is to show the effectiveness of the proposed DS-GRFT family  in a typical situation. We  consider a group of UAVs consisting of three moving targets with similar motion parameters, as shown in Table \ref{target-para}.
\begin{figure}[t!]
\begin{minipage}[c]{0.49\linewidth}
\centering
\centerline{\includegraphics[width=4.5cm]{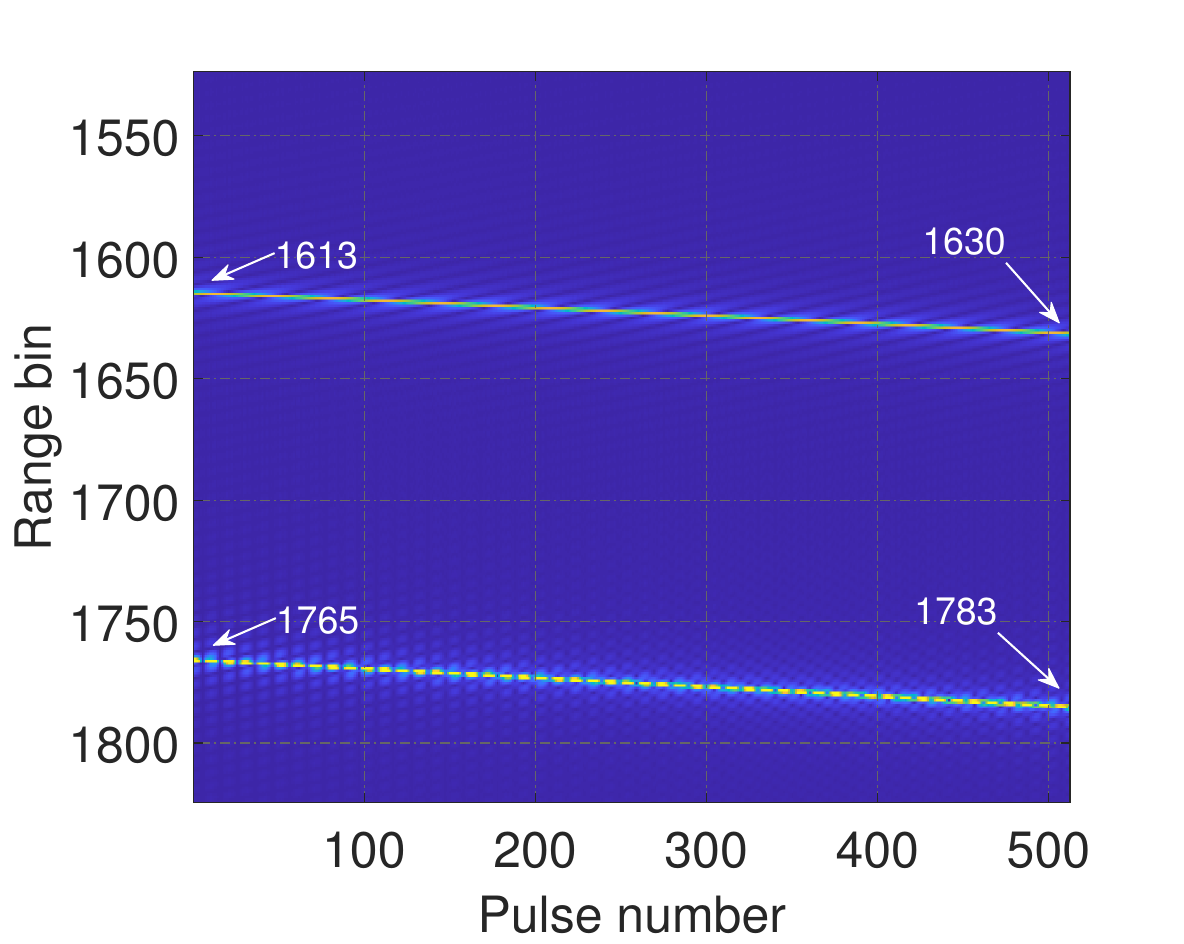}}
\centerline{\small{\small{(a)}}}
\end{minipage}
\hfill
\begin{minipage}[c]{0.49\linewidth}
\centering
\centerline{\includegraphics[width=4.5cm]{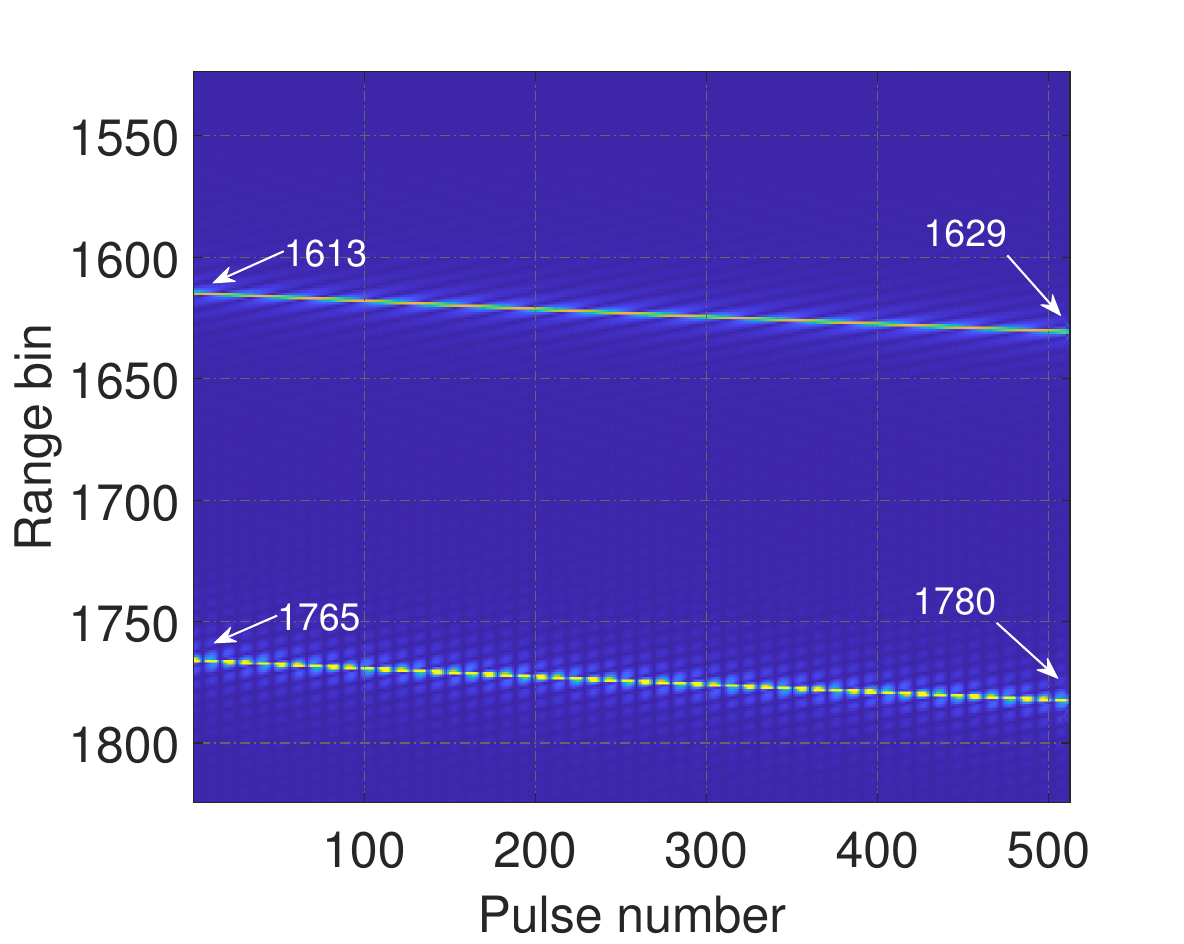}}
\centerline{\small{\small{(b)}}}
\end{minipage}
\vfill
\begin{minipage}[c]{0.49\linewidth}
\centering
 \centerline{\includegraphics[width=4.5cm]{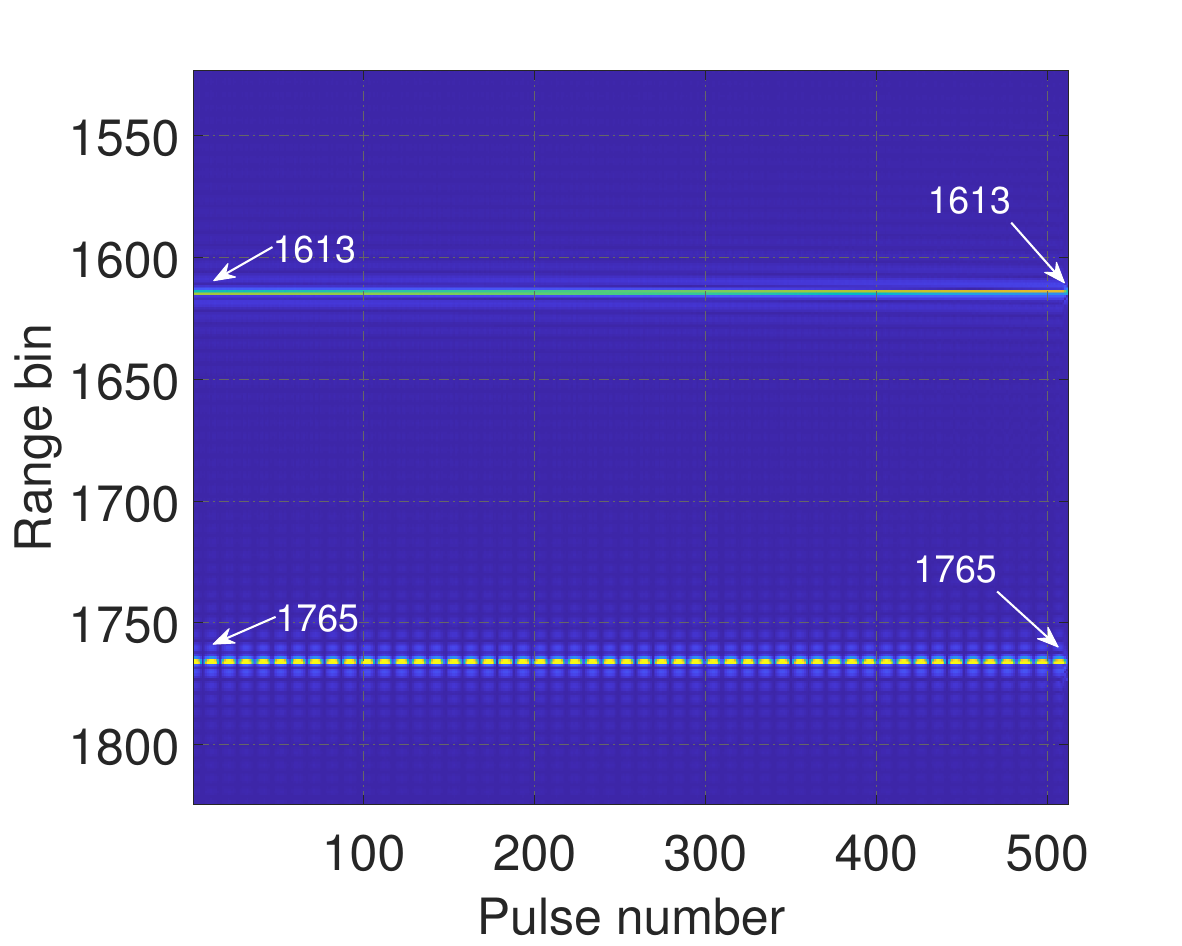}}
 \centerline{\small{\small{(c)}}}
\end{minipage}
\hfill
\begin{minipage}[c]{0.49\linewidth}
\centering
\centerline{\includegraphics[width=4.5cm]{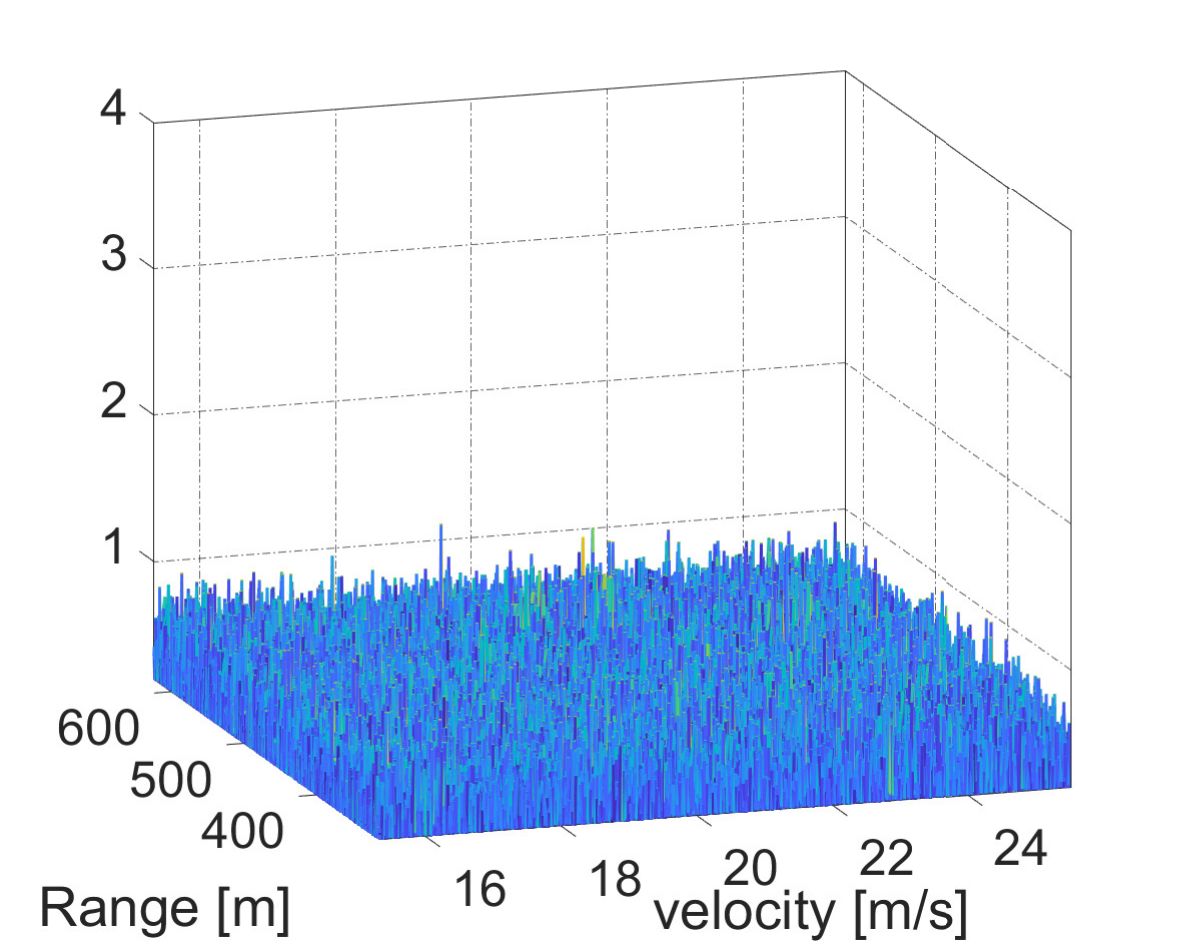}}
\centerline{\small{\small{(d)}}}
\end{minipage}
\caption{(a) Range spectrum with slow-time before RM correction; (b) range spectrum with slow-time after coarse acceleration compensation; (c) range spectrum with slow-time after coarse acceleration  and coarse velocity compensations; (d) range-Doppler spectrum after RM correction.}
\label{coarse-compensation-instance}
\end{figure}

Fig. \ref{coarse-compensation-instance} (a) shows the range-pulse spectrum  before the KT procedure.  An obvious RM effect can be observed in Fig. \ref{coarse-compensation-instance}  (a).   Fig. \ref{coarse-compensation-instance}  (b) shows the range spectrum with slow-time after KT.  It can be seen that the RM effect is partially corrected by one range bin for target 3 and three range bins for targets 1 and 2.  In Fig. \ref{coarse-compensation-instance}  (c), the result after coarse
 compensation is given, showing that the RM effect has been completely eliminated.
 The aforementioned results clearly verify Proposition 1 in the sense that the RM effect can be eliminated by  the coarse compensation only. Notice that, for the sake of demonstration, the results in Figs. \ref{coarse-compensation-instance}  (a) - (c) are given without noise background.
 However, as Fig. \ref{coarse-compensation-instance}  (d) shows, even if the RM effect has been eliminated,  the target motion trajectory is  submerged in the noise background due to the DFM effect.
\begin{figure}[htb]
\begin{minipage}[c]{0.49\linewidth}
\centering
 \centerline{\includegraphics[width=4.5cm]{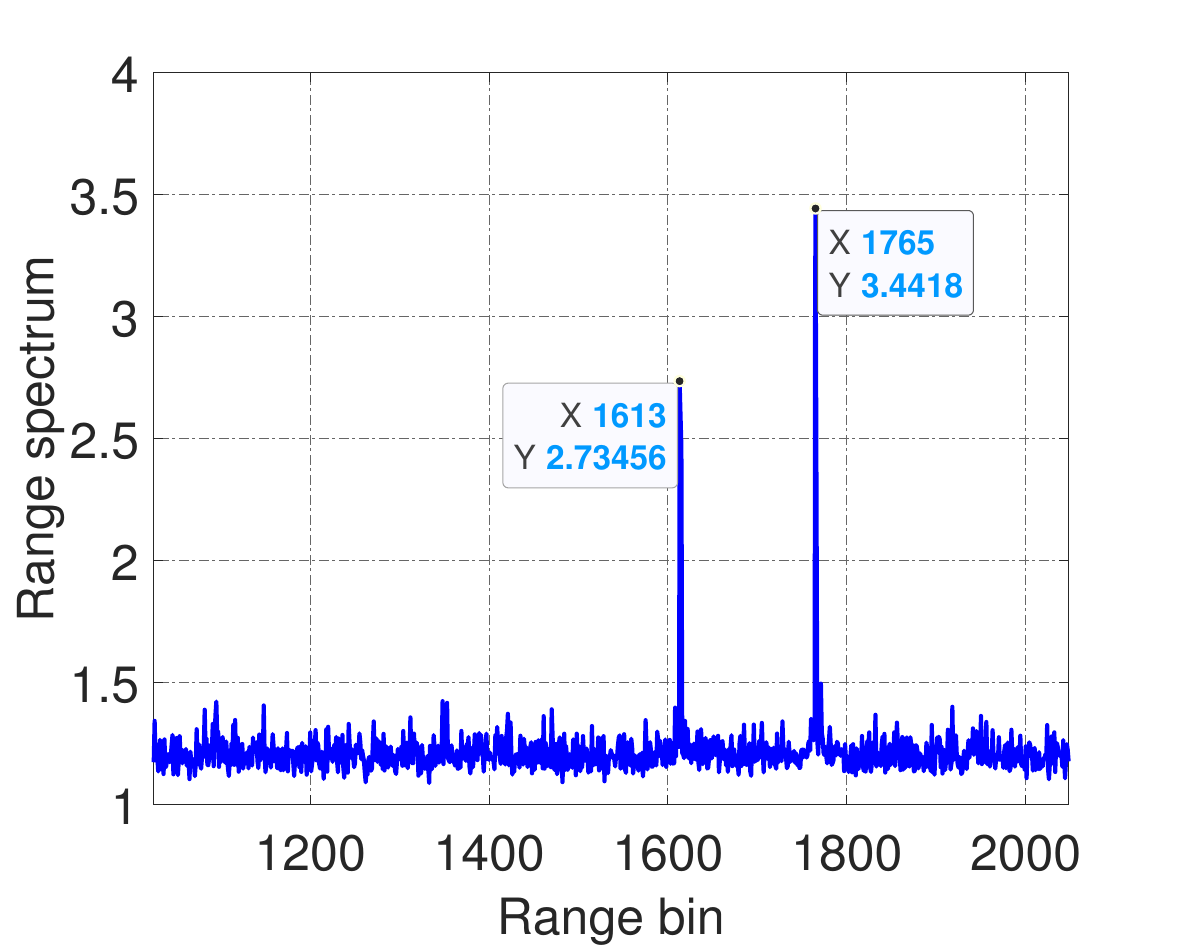}}
 \centerline{\small{\small{(a)}}}
\end{minipage}
\hfill
\begin{minipage}[c]{0.49\linewidth}
\centering
\centerline{\includegraphics[width=4.5cm]{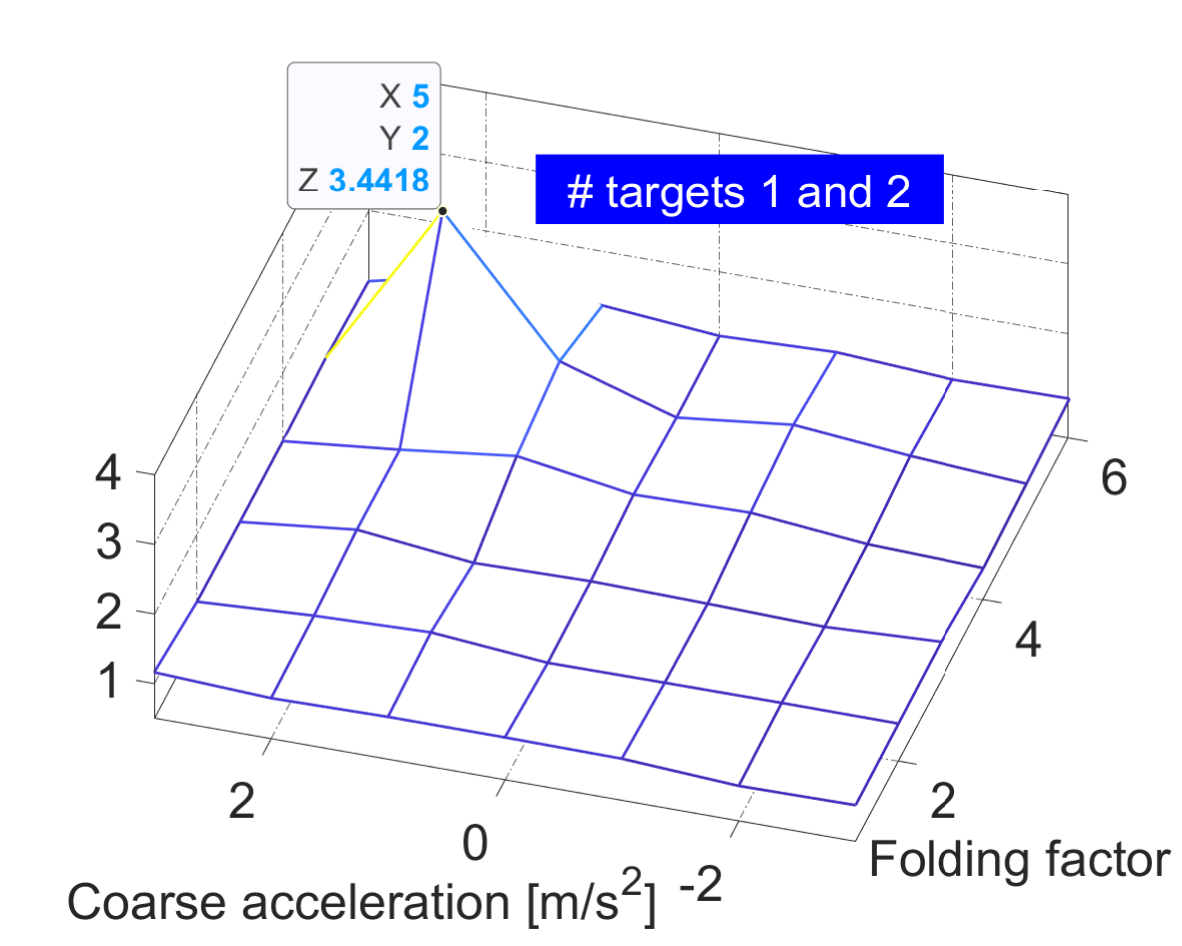}}
\centerline{\small{\small{(b)}}}
\end{minipage}
\vfill
\begin{minipage}[c]{0.49\linewidth}
 \centering
\centerline{\includegraphics[width=4.5cm]{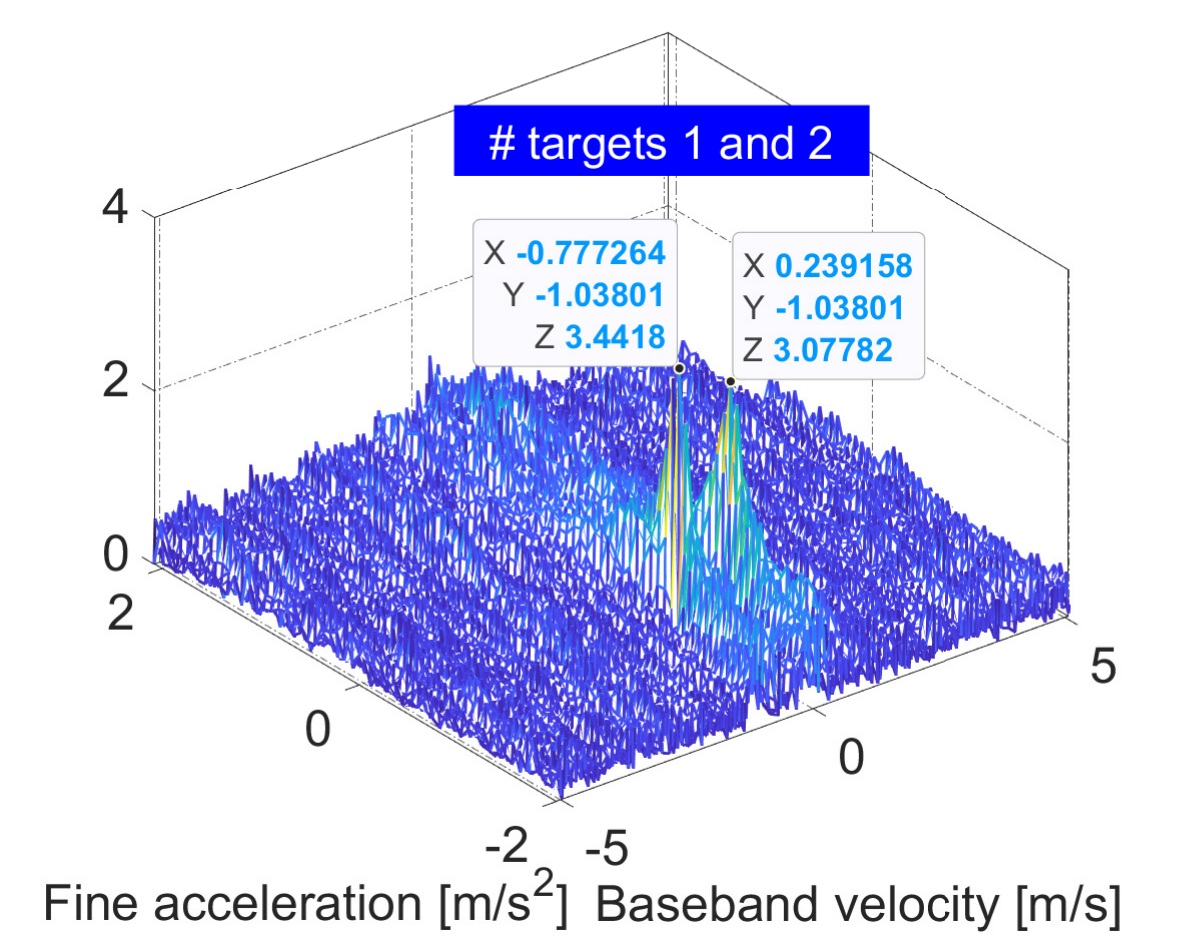}}
\centerline{\small{\small{(c)}}}
\end{minipage}
\hfill
\begin{minipage}[c]{0.49\linewidth}
 \centering
\centerline{\includegraphics[width=4.5cm]{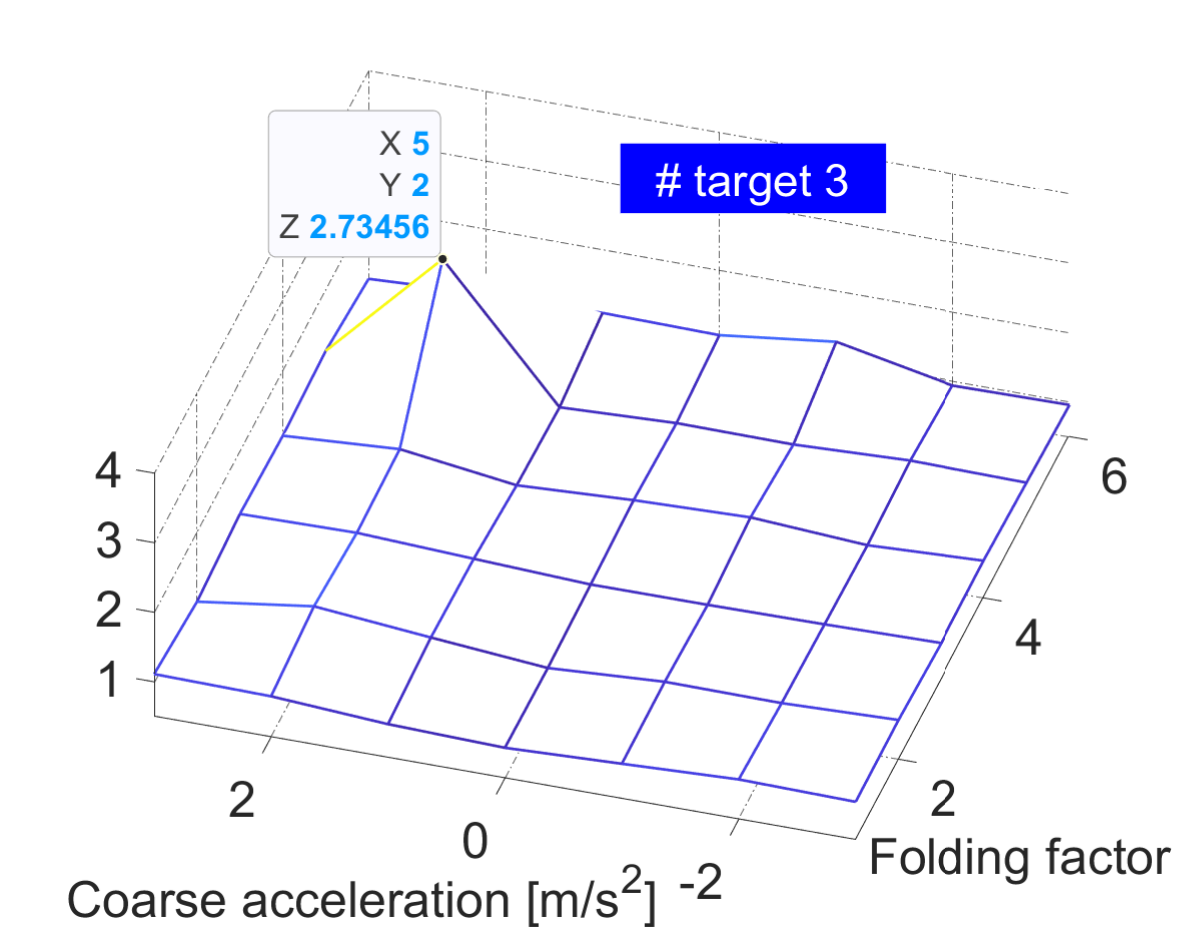}}
\centerline{\small{\small{(d)}}}
\end{minipage}
\vfill
\begin{minipage}[c]{0.49\linewidth}
 \centering
\centerline{\includegraphics[width=4.5cm]{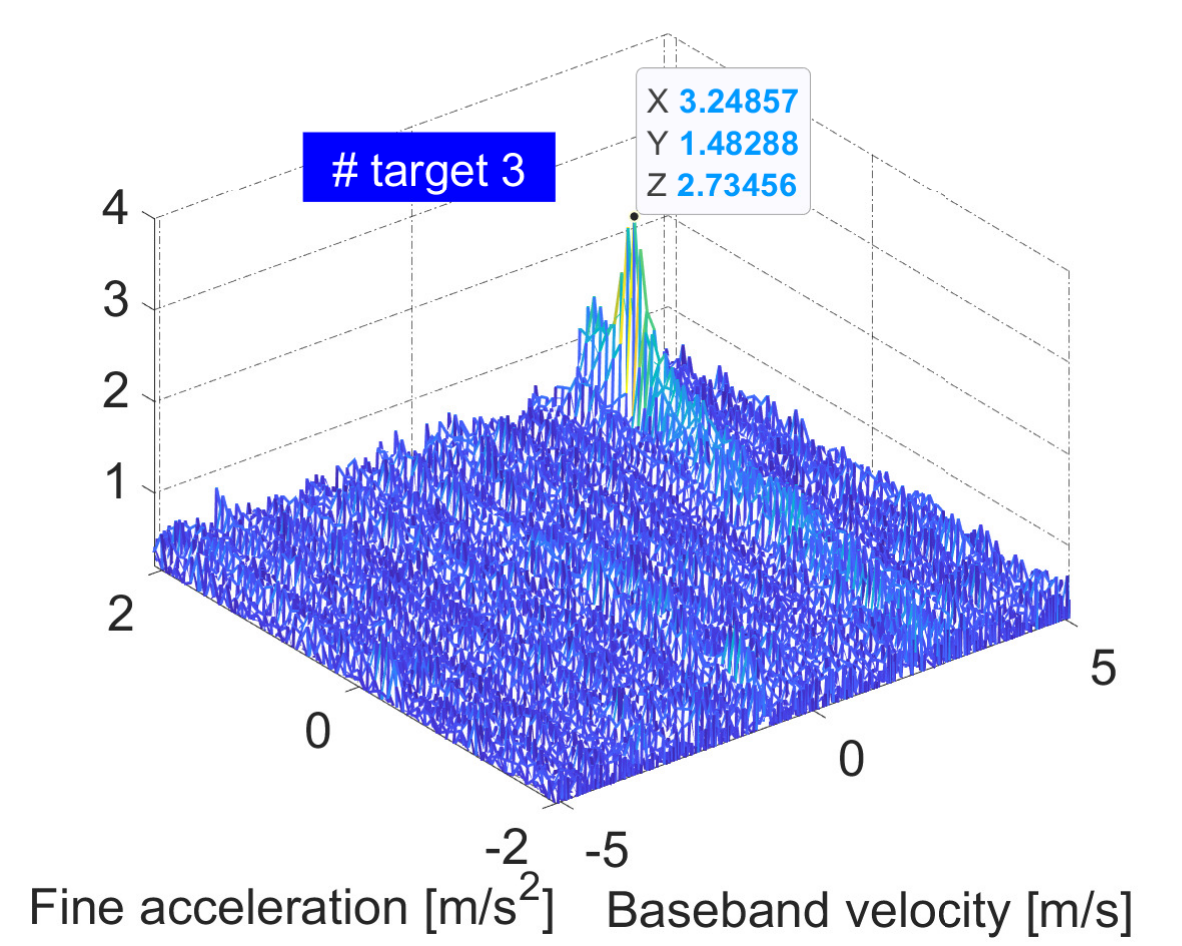}}
\centerline{\small{\small{(e)}}}
\end{minipage}
\hfill
\begin{minipage}[c]{0.49\linewidth}
 \centering
\centerline{\includegraphics[width=4.5cm]{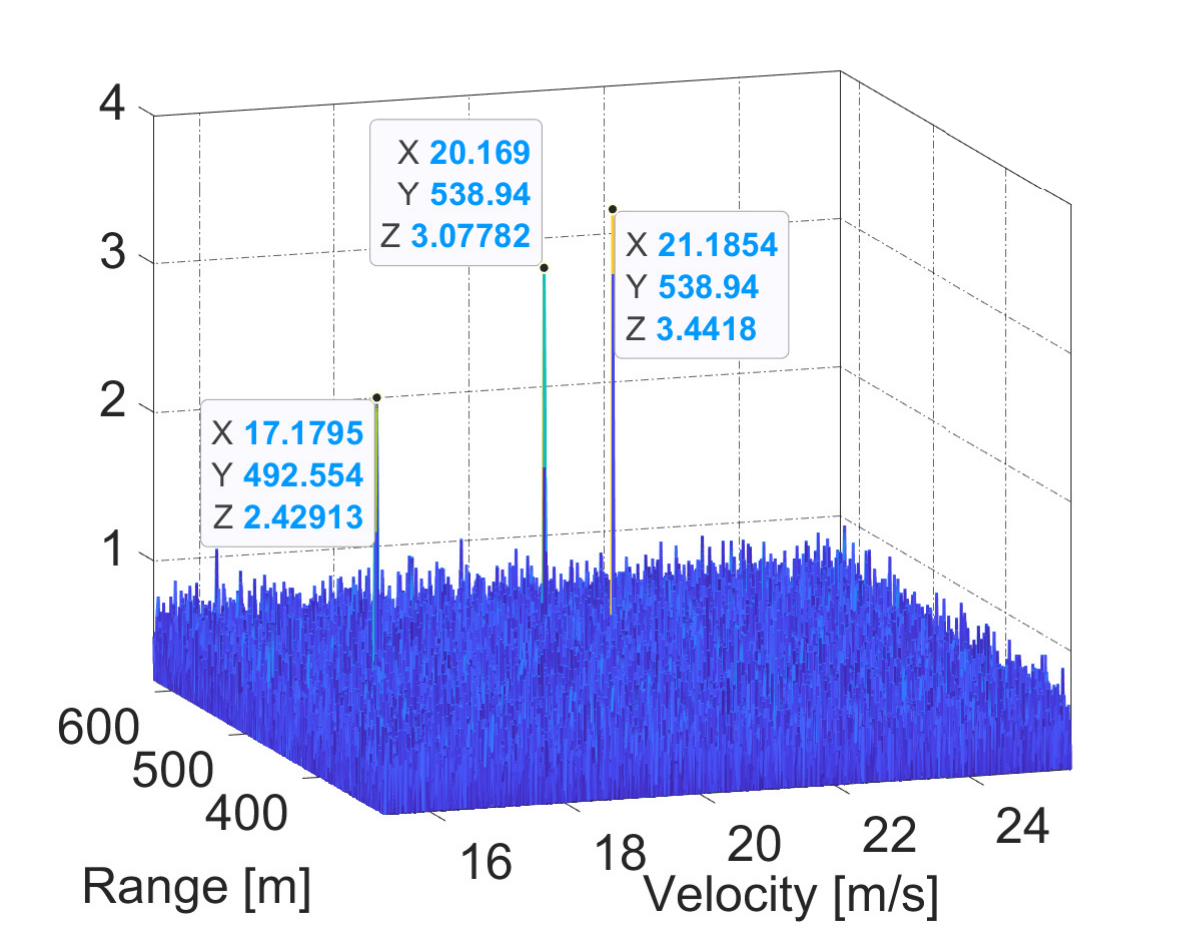}}
\centerline{\small{\small{(f)}}}
\end{minipage}
\caption{Estimation of target motion parameters:  (a) maximum output of each range bin; (b) coarse
acceleration and  folding factor spectrum at range bin 1765; (c)  fine acceleration and baseband velocity spectrum  at bin-pair $(5,2)$ of (b); (d) coarse acceleration and folding factor spectrum at range bin 1613; (e) fine acceleration and baseband velocity spectrum  at bin-pair $(5,2)$ of (d); (f) range-Doppler spectrum after target motion compensation.}
\label{Motion_Para_Est}
\end{figure}
\begin{figure}[h!]
\begin{minipage}[c]{0.49\linewidth}
\centering
 \centerline{\includegraphics[width=4.5cm]{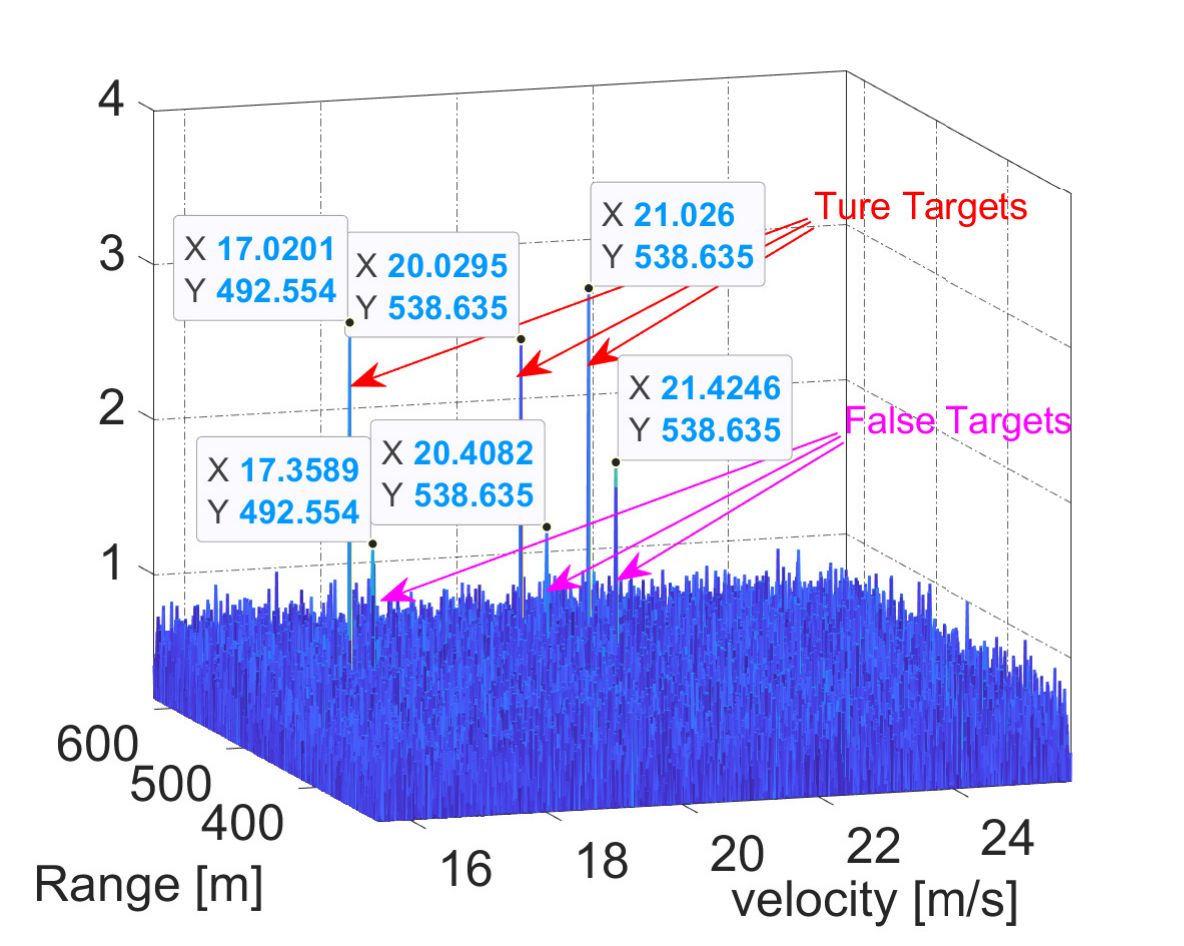}}
 \centerline{\small{\small{(a)}}}
\end{minipage}
\hfill
\begin{minipage}[c]{0.49\linewidth}
\centering
\centerline{\includegraphics[width=4.5cm]{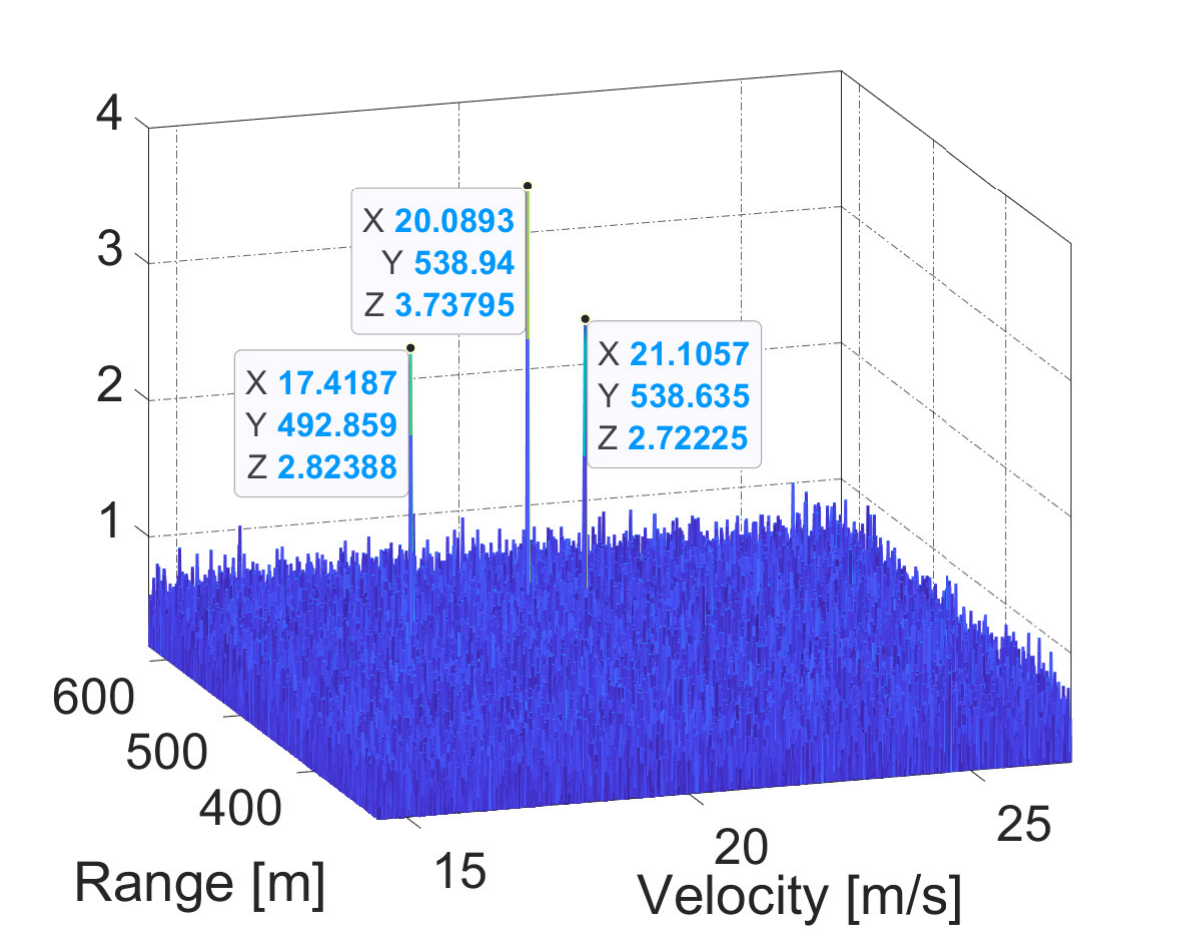}}
\centerline{\small{\small{(b)}}}
\end{minipage}
\caption{Range-Doppler spectrum  estimation: (a) TD-GRFT; (b) DS-GRFT.}
\label{RDspec_for_comparison}
\end{figure}

Fig. \ref{Motion_Para_Est}  illustrates the procedure of DS-based estimation of target motion parameters.
In particular, Fig.  \ref{Motion_Para_Est}  (a) plots the maximum outputs of DS-KT-MFP over slow time for each range bin. It can be observed that the peaks  occur at range bins
1613 (492.44 m)  and 1765 (538.5 m) which are consistent with the  slant range values of the true objects. Figs.  \ref{Motion_Para_Est}  (b) and (c) show the DS-based motion parameter estimation procedure for targets 1 and 2. Specifically, Fig.  \ref{Motion_Para_Est}  (b) shows the spectrum with respect to the folding factor and the coarse acceleration at range bin 1765 for targets 1 and 2. Observing Fig.  \ref{Motion_Para_Est}  (b), the maximum output is at bin-pair $(5,2)$.  Then, at this maximum,  we further unfold the  fine acceleration and baseband velocity spectrum as shown in Fig. \ref{Motion_Para_Est}  (c), from which  we observe that for target 1, estimations  of the fine acceleration  and the baseband velocity  are $-1.038\,\text{m}/\text{s}^2$ and $-0.777\,\text{m}/\text{s}$, respectively and for target 2, the counterparts are  $-1.038\,\text{m}/\text{s}^2$ and $0.239\,\text{m}/\text{s}$, respectively.
Similarly to Figs.  \ref{Motion_Para_Est}  (b) and (c), Figs. \ref{Motion_Para_Est}  (d) and (e) show the DS-based motion parameter estimation procedure for target 3.
Finally,  Fig. \ref{Motion_Para_Est} (f) shows the range-Doppler spectrum after both RM and DFM are compensated. It can be seen that the target power is well-focused after compensation of target motion.

Fig. \ref{RDspec_for_comparison}  further shows the range-Doppler spectra of DS-GRFT and TD-GRFT for comparison. As it can be seen from Fig. \ref{RDspec_for_comparison} (b), the proposed DS-GRFT algorithm can also provide a clear range-Doppler spectrum with three targets.
However, for the TD-GRFT algorithm in Fig. \ref{RDspec_for_comparison} (a), besides the three true targets focused in their right positions, there are also three false targets with the same range but  different velocity corresponding to the true target. The reason behind this phenomenon will be under our future investigation.
\subsection{Experiment 2}
 In this section,  the performance of the proposed   DS-GRFT and DS-KT-MFP detectors is assessed via 500 Monte Carlo runs. The TD-GRFT, FD-GRFT and KT-MFP detectors are also considered as benchmarks.
\begin{figure}[htb]
\begin{minipage}[c]{0.49\linewidth}
\centering
 \centerline{\includegraphics[width=4.5cm]{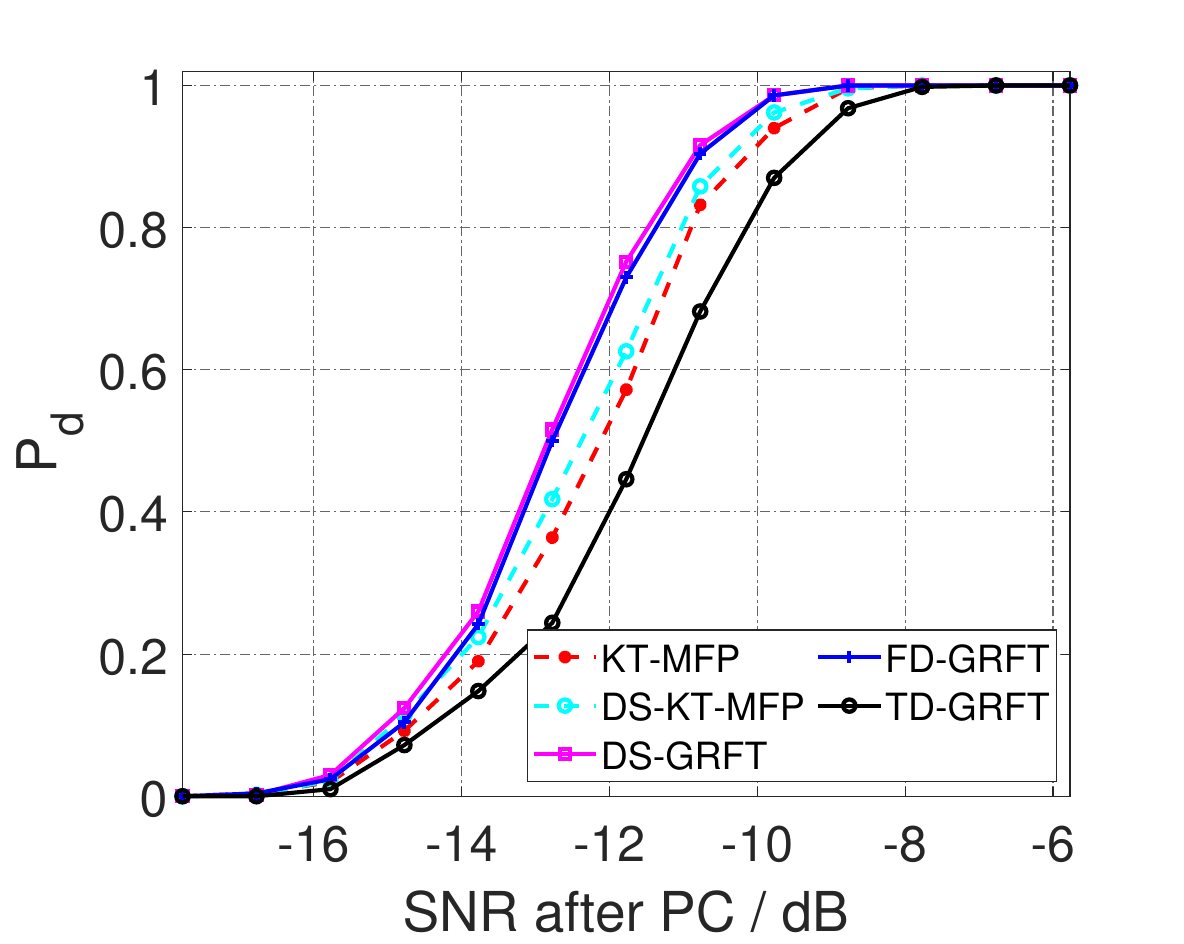}}
 \centerline{\small{\small{(a)}}}
\end{minipage}
\hfill
\begin{minipage}[c]{0.49\linewidth}
\centering
\centerline{\includegraphics[width=4.5cm]{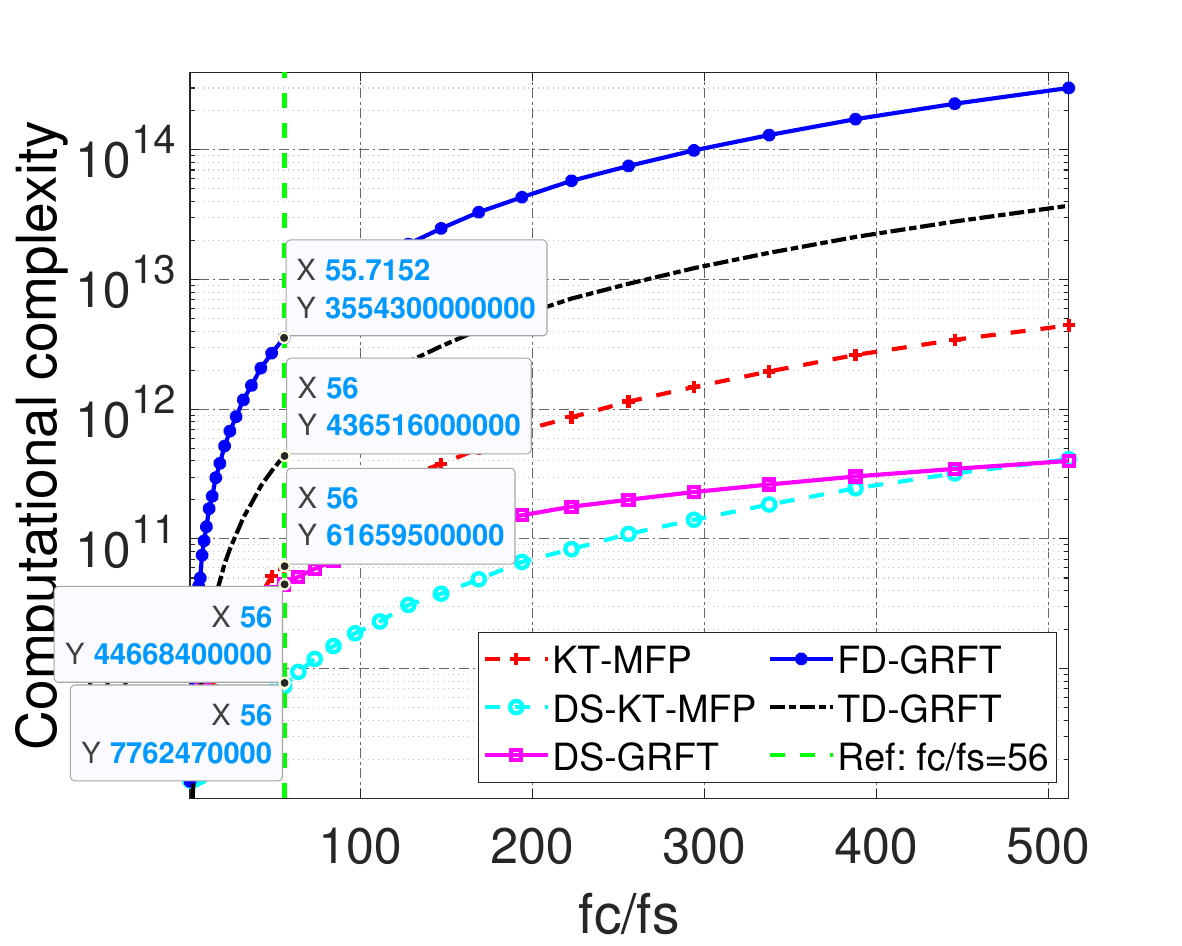}}
\centerline{\small{\small{(b)}}}
\end{minipage}
\caption{Performance analysis  for  different detectors: (a) Detection probability versus pulse-compressed SNR; (b) Computational complexity versus $\frac{f_c}{f_s}$.}
\label{Performance-analysis}
\end{figure}
Fig. \ref{Performance-analysis} (a)  shows the curves of detection performance for the  considered five methods in different input SNR cases.
 It can be seen that the proposed DS-KT-MFP algorithm is slightly better than KT-MFP, while the proposed DS-GRFT  and GRFT  detectors achieve almost the same performance. These results demonstrate the effectiveness of the proposed DS-based LTCI  algorithms.

In addition, it can also be seen from Fig. \ref{Performance-analysis} (a) that the detection performance of the  TD-GRFT detector is worse than both  FD-GRFT  and DS-GRFT detectors. Specifically, the former has $1.5$ dB performance loss at $P_d=0.9$.  This result confirms the theoretical analysis of the TD-GRFT detector in Section III-A.
Moreover, compared to the proposed DS-GRFT detector, the proposed DS-KT-MFT detector has $0.7$ dB performance loss at the detection probability $P_d=0.9$. This is because of the sinc-like interpolation for the KT procedure.  

To evaluate the computation reduction of the proposed algorithms in a comprehensive manner, we further show how the computational complexity  varies with  the ratio $\frac{f_c}{f_s}$  for different algorithms in Fig. \ref{Performance-analysis} (b).
The velocity scope is set to $-50 \sim 50$ $\text{m}/\text{s}$, and the acceleration scope to $-30 \sim 30$ $\text{m}/\text{s}^2$.
It can be seen from Fig. \ref{Performance-analysis} (b) that the DS-GRFT detector offers a reduction of computational cost by over 1-2 orders of magnitude compared to the  FD-GRFT or TD-GRFT detector; while  the reduction of computational cost of DS-KT-MFP can be up to 1 order of magnitude compared to KT-MFP. Notice that when  $\frac{f_c}{f_s}\geq 48$,  the computational cost of DS-GRFT  is lower than KT-MFP, while  when $\frac{f_c}{f_s}$ approaches $512$, the computational cost of  DS-GRFT is close to DS-KT-MFP.

Then Table \ref{Execution-Time-different-algorithm} reports the execution times of different considered detectors at $\frac{f_c}{f_s}=56$.  Experiments are carried out in MATLAB 2019b using  Intel Core i7-11700 with an 8-core 2.5-GHz CPU and 32 GB of RAM. It can be seen that DS-KT-MFP takes $41.1$s, while KT-MFP needs almost three times as long;  DS-GRFT takes $141.1$s, while FD-GRFT needs $24649$s which is almost 174 times of the former. Furthermore, although DS-KT-MFP yields  $0.7$dB detection performance loss compared to DS-GRFT,  the execution time of the former  accounts for less time.
Notice that the execution time reduction of the DS-based detector implemented in MATLAB is not exactly the same resulting from the theoretical analysis shown in Fig. 7, mainly due to additional efficient implementations of two-dimensional FFT  and matrix multiplication  which are built-in operations  in MATLAB.

 To summarize, the above results  demonstrate   the computational efficiency of the proposed  DS-based detectors  as well as comparable performance with respect to the  GRFT detector family.
\begin{table}[htb]
\renewcommand{\arraystretch}{2}
\caption{Execution times of different detectors in Experiment 2.}\label{execution-time}
\begin{center}
\centering
\footnotesize
\begin{tabular*}{0.48\textwidth}{@{\extracolsep{\fill}}c|ccccc}
\hline\hline
Alg.         &  TD-GRFT&  GRFT & DS-GRFT &KT-MFP&  DS-KT-MFP   \\
\hline
\makecell{Execution \\Time [s]}   & 8400& 24649 & 141.7&119.8  & 41.1\\
\hline
\end{tabular*}
\label{Execution-Time-different-algorithm}
\end{center}
\end{table}

\section{Conclusion}
This paper  has proposed a dual-scale decomposition of the target motion parameters,  which allows for reducing the standard GRFT family (including FD-GRFT and KT-MFP) into a GIFT process in range domain and GFT processes in Doppler domain conditioned on the coarse motion parameters. Thanks to this appealing property, the joint correction of RM and DFM effects is decoupled into a cascade procedure, consisting of RM compensation on the coarse search space followed by DFM compensation on the fine search spaces.  Compared to the standard GRFT family, the proposed dual-scale GRFT family can provide comparable performance but with significant improvement in computational efficiency. Simulation experiments demonstrate the performance gain of the proposed method.

\bibliographystyle{IEEEtran}
\bibliography{DS-GRFT}

\end{document}